\documentclass[amsmath,  floatfix, nofootinbib, prx, aps, usenames, dvipsnames,10pt, allowfontchageintitle, a4paper]{quantumarticle_equalcontrib}

\usepackage{hyperref} 
\usepackage{jabbrv}

\usepackage{enumitem}
\let\olddesc\description
\def\description{\olddesc\setlist[itemize]{leftmargin=*,labelindent=-12pt}}

\usepackage{lineno}

\usepackage{xr}

\usepackage[english]{babel}

\usepackage{graphicx}
\usepackage{caption}
\usepackage{subcaption}

\usepackage{xcolor}

\usepackage{tcolorbox}
\tcbuselibrary{breakable}
\tcbuselibrary{skins}
\tcbset{enhanced, colback=white, arc=0mm, outer arc=0mm, fonttitle=\bfseries, fontupper=\small}
\newtcolorbox[auto counter]{examplebox}[2][]{%
title=Box~\thetcbcounter: #2,#1}

\usepackage{amsmath}
\usepackage{amsthm}
\usepackage{amssymb}
\usepackage{mathtools}
\usepackage{thmtools}

\usepackage{bm}  
\usepackage{bbm}
\usepackage{dsfont}
\usepackage{cases}
\usepackage{nicefrac} 
\usepackage{xfrac}

\usepackage{multicol}

\newtheorem{theorem}{Theorem}

\newcommand{\ket}[1]{\left| #1 \right>} 
\newcommand{\bra}[1]{\left< #1 \right|} 

\let\epsilon\relax
\def\epsilon{\varepsilon}
\let\phi\relax
\def\phi{\varphi}
\def\R{\mathbb{R}}

\def\id{\mathbbm 1}

\begin{document}

\title{Operationally meaningful representations \\ of physical systems in neural networks}

\author{Hendrik Poulsen Nautrup\**\,}
\email{hendrik.poulsen-nautrup@uibk.ac.at}
\affiliation{Institute for Theoretical Physics, University of Innsbruck, Technikerstr. 21a, A-6020 Innsbruck, Austria}

\author{Tony Metger\**\,}
\email{tmetger@ethz.ch}
\affiliation{Institute for Theoretical Physics, ETH Z{\"u}rich, 8093 Z{\"u}rich, Switzerland}

\author{Raban Iten\**\,}
\affiliation{Institute for Theoretical Physics, ETH Z{\"u}rich, 8093 Z{\"u}rich, Switzerland}

\author{Sofiene Jerbi}
\author{Lea M. Trenkwalder}
\affiliation{Institute for Theoretical Physics, University of Innsbruck, Technikerstr. 21a, A-6020 Innsbruck, Austria}

\author{Henrik Wilming}
\affiliation{Institute for Theoretical Physics, ETH Z{\"u}rich, 8093 Z{\"u}rich, Switzerland}

\author{Hans J. Briegel}
\affiliation{Institute for Theoretical Physics, University of Innsbruck, Technikerstr. 21a, A-6020 Innsbruck, Austria}
\affiliation{Department of Philosophy, University of Konstanz, 78457 Konstanz, Germany}

\author{Renato Renner}
\affiliation{Institute for Theoretical Physics, ETH Z{\"u}rich, 8093 Z{\"u}rich, Switzerland}

\date{}

\begin{abstract} 

\noindent 
To make progress in science, we often build abstract representations of physical systems that meaningfully encode information about the systems.
The representations learnt by most current machine learning techniques reflect statistical structure present in the training data;
however, these methods do not allow us to specify explicit and operationally meaningful requirements on the representation.
Here, we present a neural network architecture based on the notion that agents dealing with different aspects of a physical system should be able to communicate relevant information as efficiently as possible to one another. This produces representations that separate different parameters which are useful for making statements about the physical system in different experimental settings.
We present examples involving both classical and quantum physics.
For instance, our architecture finds a compact representation of an arbitrary two-qubit system that separates local parameters from parameters describing quantum correlations.
We further show that this method can be combined with reinforcement learning to enable representation learning within interactive scenarios where agents need to explore experimental settings to identify relevant variables.
\end{abstract} 
 
\maketitle

\section{Introduction}

Neural networks are among the most versatile and successful tools in machine learning~\cite{nielsenneural,lecun_deep_2015,silver_mastering_2016} and have been applied to a wide variety of problems in physics (see~\cite{dunjko_machine_2018,roscher_explainable_2019,carleo_machine_2019} for recent reviews). Many of the earlier applications have focused on solving specific problems that are intractable analytically and for which conventional numerical methods deliver only unsatisfactory results. Conversely, neural networks may also lead to new insights into how the human brain develops physical intuition from observations~\cite{Bates_humans_2015,Wu_galileo_2015, Bramley2018,rempe_learning_2019,kissner_adding_2019,ehrhardt_unsupervised_2018,ye_interpretable_2018,zheng_unsupervised_2018}.   

Recently, the potential role that machine learning might play in the scientific discovery process has received increasing attention~\cite{raban_2018_discovering,melnikov_active_2018,ried_2019_how,briegel_2012_on, wu_toward_2018,de_simone_guiding_2019, dagnolo_learning_2019,rahaman_2019_learning}.
This direction of research is not only concerned with machine learning as a useful numerical tool for solving hard problems, but also seeks ways to establish artificial intelligence methodologies as general-purpose tools for scientific research. This is motivated from various directions: from an artificial intelligence perspective, having machines autonomously discover scientific concepts about the world is often seen as an important step towards artificial general intelligence~\cite{Lake2016}; from the perspective of science, machine learning might complement human scientific research to both speed up scientific discovery and make it less susceptible to human biases.

An important step in the scientific process is to convert experimental data, which can be seen as a very high-dimensional and noisy representation of a physical system, to a more succinct representation that is amenable to a theoretical treatment. For example, when we observe the trajectory of an object, the natural experiment is to record the position of the object at different times; however, our theories of kinematics do not use time series of positions as variables, but rather describe the system using quantities, or \emph{parameters}, such as velocity and initial position.
Concepts such as velocity are more versatile because they can be used in different ways for making predictions in many different physical settings.

\begin{figure*}[ht!] 
\centering
\includegraphics[width=0.75\textwidth]{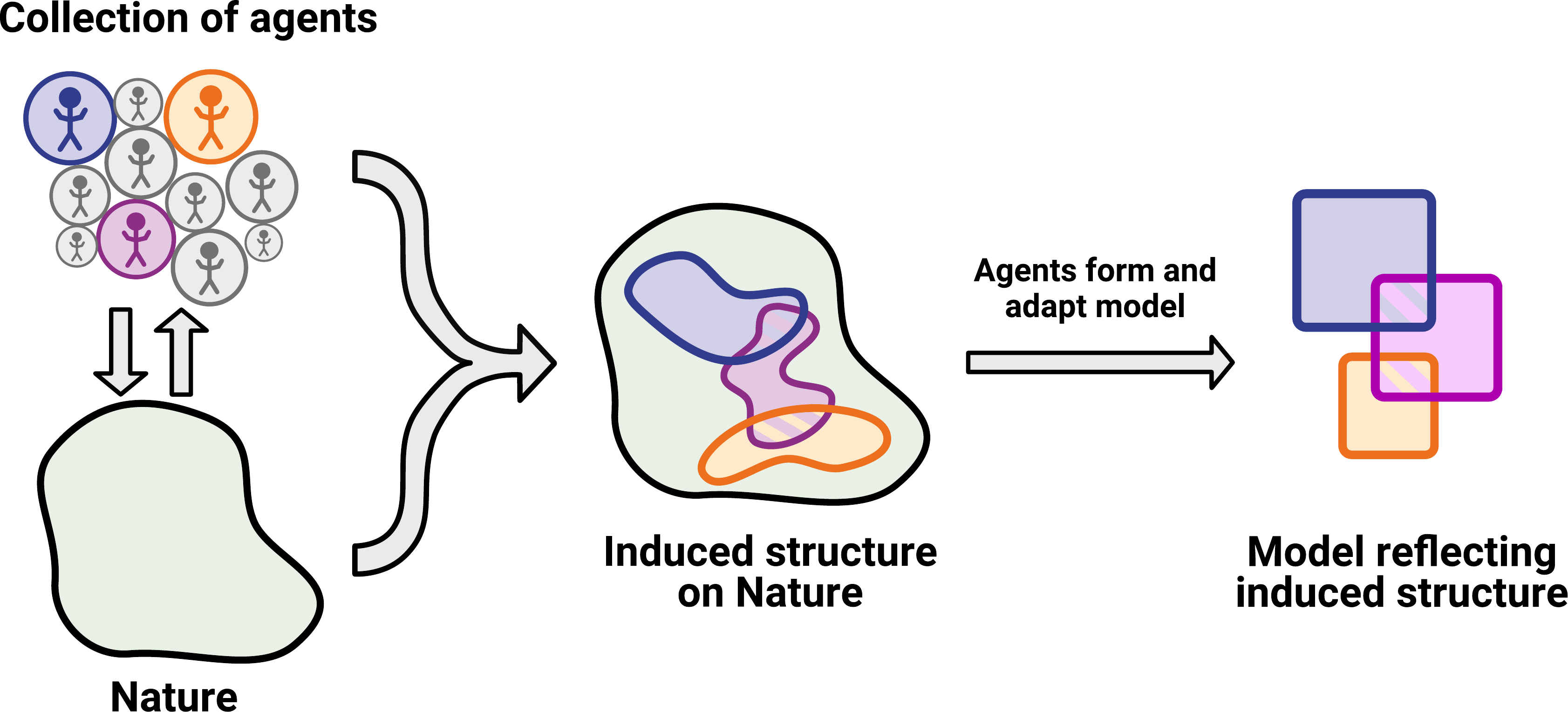} 
\caption{ 
\textbf{Conceptual overview.} 
Various agents interact with different aspects of nature through experiments. We assume that different agents only deal with parts of the experimental settings to achieve their objective. This induces structure on our description of Nature. When the agents build a model of Nature, i.e., learn to parameterise their experimental settings, we want this model to reflect the structure enforced by the requirement that agents compress and communicate parameters in the most efficient way.
}  
\label{fig:overview}
\end{figure*}

When using neural networks to find such parameterisations, one encounters the limitation of standard techniques from representation learning~\cite{Bengio2012,hinton_reducing_2006,Higgins2017},  an area of machine learning devoted to problems of this type. 
With these standard techniques, we are typically not able to specify explicit criteria on the parameterisation, such as which aspects of a system should be stored in distinct parameters. 
Instead, a separation or \emph{disentanglement} typically arises implicitly from the statistical distribution of the training data set. 
This works well for many practical problems~\cite{Higgins2017}; however, for scientific applications, it is desirable that different parameters in the representation are relevant for different experiments one can actually perform on the system.  
Otherwise it is likely that our model reflects biases we implicitly, and likely unknowingly, had in collecting the experimental data.
In the following we will call a representation that fulfills this desideratum an operationally meaningful representation.

Naturally, formulating operationally meaningful requirements for a representation and translating them to a neural network implementation depends heavily on the specific scenario one is interested in. In this work, we consider a scenario which is particularly relevant in the context of scientific discovery. Specifically, we impose structure on the parameterisation of experimental data\footnote{We consider a specific notion of measurement which we elaborate on in Sec.~\ref{sec:formal_setting}.} by assuming that different agents, which deal with different sets of questions, each only require knowledge of a subset of the parameters to successfully answer any specific question from their respective set of questions.
For instance, two agents may each have to predict the movement of a charged particle in the presence of an electrostatic field: the field's strength is relevant for both agents, whereas the individual parameters of each particle, such as charge, mass, etc., are only relevant for one agent; or many agents have to make sequences of operations to answer whether (and if so, how) various phenomena -- such as high-dimensional entanglement~\cite{melnikov_active_2018} -- can be generated in an experiment.

More generally, the criterion for imposing structure can be understood in terms of communicating agents.\footnote{Throughout this paper, we use the term \emph{agent} in a generic sense~\cite{RusselNorvig2003} and do not specifically mean reinforcement learning agents.}  An ensemble of agents would like to predict the results of various experimental settings. However, only one agent \textit{A} has access to reference data from the experimental setting (which is the high-dimensional full representation of the experimental setting). 
This agent has to identify and communicate the relevant parameters of the system to the other agents, each of whom only requires partial information to solve their question. 
Agent \textit{A} therefore splits the parameters in such a way that the remaining agents can share them optimally, in the sense that each agent requires the smallest possible subset of parameters and that parameters required by multiple agents are shared without redundancies. We formalise this notion in Sec.~\ref{sec:formal_setting}.  

In this work, we introduce a network architecture that allows us to explicitly impose the aforementioned operational criterion on the parameters used by a neural network to represent a physical system~\cite{raban_2018_discovering}. The model architecture is detailed in Sec.~\ref{sec:model}. 
In Sec.~\ref{sec:experiments} we provide two illustrative examples of a scenario in which agents are given (a high-dimensional representation of) an experimental setting and are required to make a prediction w.r.t. a specific question. In an example from classical mechanics, the network autonomously distinguishes parameters that are only relevant to predict the behaviour of an individual particle from parameters that affect the interaction between particles. 
This structure arises naturally in an experiment with multiple charged particles, when different agents have to predict the motion of \emph{their} charged particle in the presence or absence of the other agents' charged particles. The method is agnostic to the theory underlying an experiment and can thus also be applied to quantum mechanical experiments. We illustrate this by learning a representation of a two-qubit system that separates parameters relevant for two \textit{individual} qubits from those parameters describing the quantum correlations between qubits. In fact, this parameterisation is similar to the standard, analytic representation described in Refs.~\cite{Gamel_2016,garon_2015_visualizing}. 

In Sec.~\ref{sec:reinforcement_learning}, we consider scenarios where the answer to a specific question can be described as a sequence of actions; such a sequence either does or does not achieve a specific goal, i.e., feedback about the quality of an action may be discrete and delayed. 
For instance, this tends to be the case for optimisation problems such as the design and control of complex systems or the development of gadgets and software solutions for different scientific and technological purposes~\cite{patnaik_2013_modeling}. 
In the context of scientific discovery, the specific goal may be to build experimental settings which bring about a specific phenomenon, e.g., entanglement~\cite{melnikov_active_2018}.
In such a scenario, we may first explore the space of experimental settings and learn solutions through reinforcement learning~\cite{sutton_1998_reinforcement} before applying the criterion of minimal communication to impose structure on the parameterisation of experimental data. Therefore, we provide a formal description of reinforcement learning environments where our architecture may capture operationally meaningful structure, and demonstrate this by means of an illustrative example.

\section{Related work} \label{sec:related_work}

The field of representation learning is concerned with feature detection in raw data.
While, in principle, all deep neural network architectures learn some representation within their hidden layers, most work in representation learning is dedicated to defining and finding \textit{good} representations~\cite{Bengio2012}. 
A desirable feature of such representations is the interpretability of its parameters (stored in different neurons in a neural network). 
Standard autoencoders, for instance, are neural networks which compress data during the learning process. In the resulting representation, different parameters in the representation are often highly correlated and do not have a straightforward interpretation.
A lot of work in representation learning has recently been devoted to \emph{disentangling} such representations in a meaningful way (see e.g. \cite{Higgins2017,chen_2018_isolating,Kim2018,thomas_2018_disentangling,
francois_lavet_combined_2018}). 
In particular, these works introduce criteria, also referred to as \textit{priors} in representation learning, by which we can disentangle representations. 

\textbf{$\beta$-variational autoencoders.}
Autoencoders are one particular architecture used in the field of representation learning, whose goal is to map a high-dimensional input vector $x$ to a lower-dimensional latent vector $z$ using an encoding mapping $E(x)=z$. For autoencoders, $z$ should still contain all information about $x$, i.e., it should be possible to reconstruct the input vector $x$ by applying a decoding function $D$ to $z$. The encoder $E$ and the decoder $D$ can be implemented using neural networks and trained unsupervised by requiring $D(E(x))=x$.
$\beta$-variational autoencoders ($\beta$-VAEs) are autoencoders where the encoding is regularised in order to capture statistically independent features of the input data in separate parameters~\cite{Higgins2017}.  

In Ref.~\cite{raban_2018_discovering} a modified $\beta$-VAE, called \emph{SciNet}, was used to answer questions about a physical system. The criterion by which the latent representation is disentangled is statistical independence equivalent to standard $\beta$-VAE methods. In the present work, we use a similar architecture but impose an operational criterion in terms of communicating agents for the disentanglement of parameters.

Another prior that was recently proposed to disentangle a latent representation is the \textit{consciousness prior}~\cite{bengio_2017_consciousness}.  There, the author suggests to disentangle abstract representations via an attention mechanism by assuming that, at any given time, only a few internal features or concepts are sufficient to make a useful statement about reality.

\textbf{State Representation Learning.}
State representation learning (SRL) is a branch of representation learning for interactive problems~\cite{lesort_2018_state}. 
For instance, in reinforcement learning~\cite{sutton_1998_reinforcement} it can be used to capture the variation in an environment created by an agent's action~\cite{bengio_2017_independently,thomas_2018_disentangling,francois_lavet_combined_2018, jonschkowski_2015_learning}. 
In Ref.~\cite{thomas_2018_disentangling}  the representation is disentangled by an \textit{independence prior} which encourages that  independently controllable features of the environment are stored in separate parameters. 
A similar approach was recently introduced in Ref.~\cite{francois_lavet_combined_2018} where model-based and model-free reinforcement learning are combined to jointly infer a sufficient representation of the environment. 
The abstract representation becomes expressive by introducing \textit{representation} and \textit{interpretability priors}. Similarly, in Ref.~\cite{jonschkowski_2015_learning} \textit{robotic priors} are introduced to impose a structure reflecting the changes that occur in the world and in the way a robot can interact with it.  
As shown in Ref.~\cite{francois_lavet_combined_2018} and~\cite{jonschkowski_2015_learning}, such requirements can lead to very natural representations in certain scenarios such as creating an abstract representation of a labyrinth or other navigation tasks. 

In Ref.~\cite{jaderberg_2017_reinforcement} many reinforcement learning agents with different tasks share a common representation which is being developed during training. They demonstrate that learning auxiliary tasks can help agents to improve learning of the overall objective. One important auxiliary task is given by a \textit{feature control prior} where the goal is to maximise the activations of hidden neurons in an agent's neural network as they may represent task-relevant high-level features~\cite{mnih_2015_human, zahavy_2016_graying}. However, this representation is not expressive or interpretable to the human eye since there is no criterion for disentanglement.

\textbf{Projective Simulation}
The projective simulation (PS) model for artificial intelligence~\cite{briegel_2012_projective} is a model for agency which employs a specific form of an episodic and compositional memory to make decisions. It has found applications in various areas of science, from quantum physics~\cite{melnikov_active_2018,nautrup_2018_optimizing,wallnofer_2019_machine} to robotics~\cite{hangl_2016_robotic,hangl_2017_skill} and the modelling of animal behaviour~\cite{ried_2019_modelling}. Its memory consists of a network of so-called \emph{clips} which can represent basic episodic experiences as well as abstract concepts. Besides the usage for generalisation~\cite{melnikov_2017_projective,falmini_2019_photonic}, these clip networks have already been used to represent abstract concepts in specific settings~\cite{ried_2019_how,hangl_2017_skill}. 
In Ref.~\cite{ried_2019_how}, PS was used to infer the existence of unobserved variables such as mass, charge or size which make an object respond in certain experimental settings in different ways. In this context, the authors point out the significance of exploration when considering the design of experiments, and thereby adopt the notion of reinforcement learning similar to Ref.~\cite{melnikov_active_2018}. In line with previous works, we will also discuss reinforcement learning methods for the design of experimental settings. Unlike previous works however, we provide an interpretation and formal description of decision processes which are specifically amenable to representation learning. Moreover, we employ neural networks architectures to infer continuous parameters from experimental data. In contrast, PS is inherently discrete and therefore better suited to infer high-level concepts.
 
In this work, we suggest to disentangle a latent representation of a neural network according to an operationally meaningful principle, by which agents should communicate as efficiently as possible to share relevant information to solve their tasks. Technically, we disentangle the representation according to different questions or tasks, as described in more detail in the following section.

\section{Formal setting}\label{sec:formal_setting}
\begin{figure*}[ht!] 
\centering
  \begin{subfigure}{0.35\textwidth}
  \centering
  \vspace{0.6cm}
  \includegraphics[width=\textwidth]{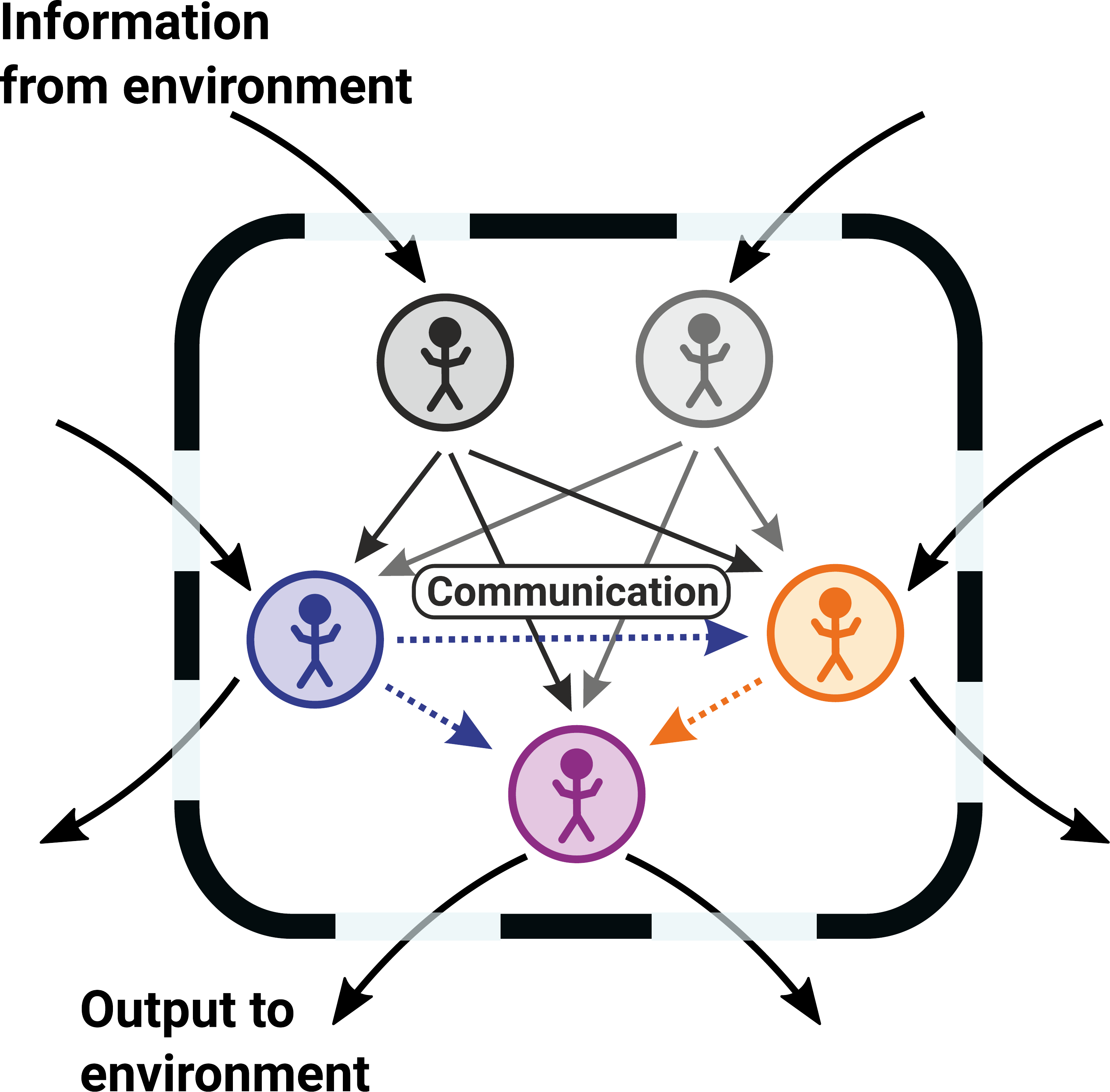} 
  \vspace{0.6cm}
  \caption{ ~ }  
\end{subfigure}
\hspace{0.05\textwidth}
  \begin{subfigure}{0.55\textwidth}
  \centering
  \includegraphics[width=\textwidth]{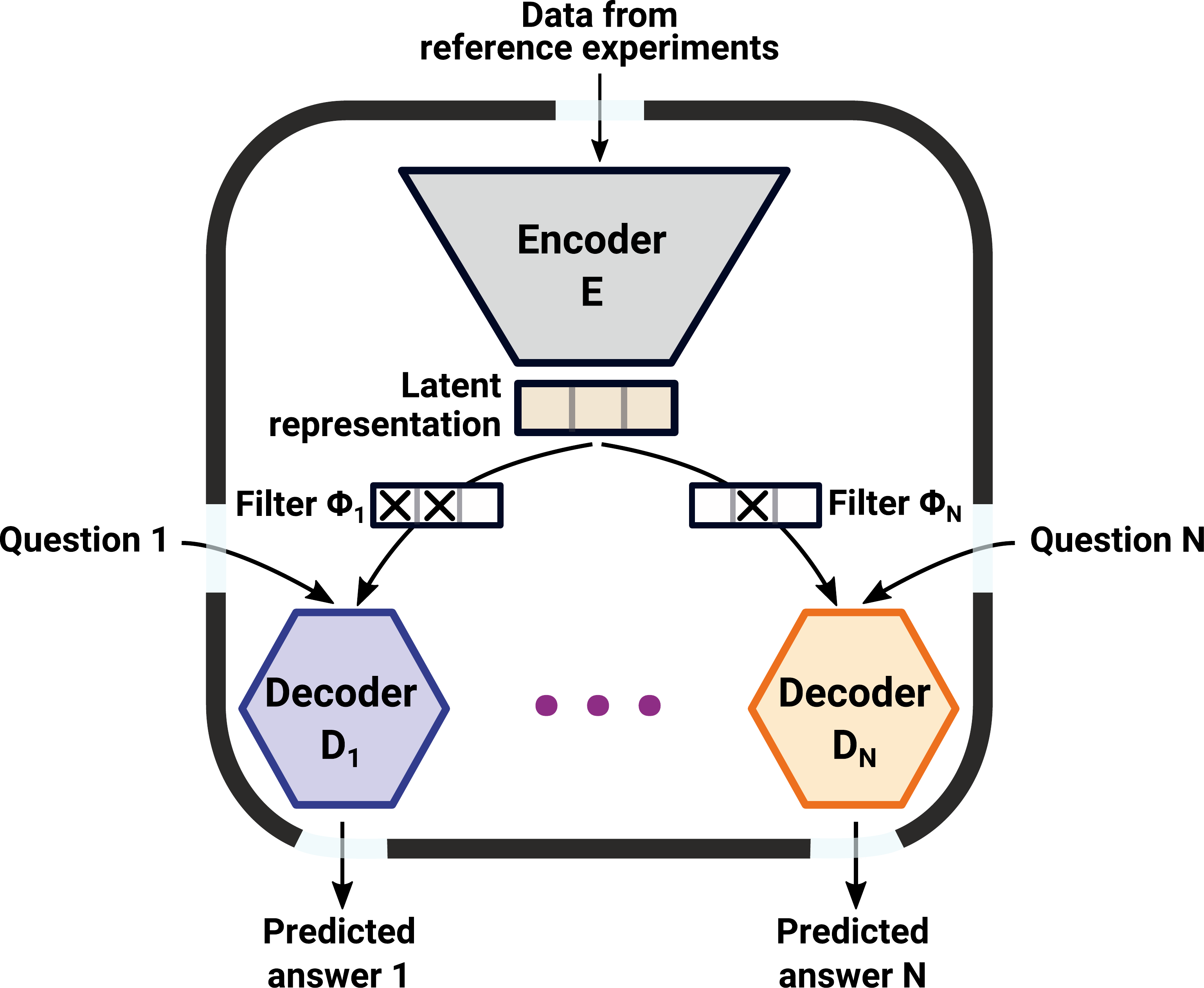} 
  \caption{ ~ }  
\end{subfigure}
\caption{\textbf{Communicating agents and an implementation with neural networks.}
(a) In the generic communication setting that we consider, agents have access to information (e.g., observations) about the environment, which can be considered as a physical system, and can interact with part of it (e.g., by making predictions). Different agents may observe or interact with different subsystems,  but might not have access to other parts of the environment. To solve certain tasks, agents may require information that can only be accessed by other agents. Therefore, the agents have to communicate. To this end, they encode their information into a representation that can be communicated efficiently to other agents, i.e., they have to find an efficient ``language" to share relevant information. Arrows passing through boundaries represent interactions with the physical system in form of data gathering or predictions. Arrows between agents suggest possible communication channels. 
(b) In the specific settings that we consider, an \emph{encoding} agent maps an observation (given as a sample) obtained from the current experimental setting onto a latent representation, part of which has to be communicated to \emph{decoding} agents.  
Decoding agents receive additional information specifying the question which they are required to answer; for example, the question may just define the problem setting. In our architecture, we thus view the question as coming from the environment, but in general, it could also include information from other agents. The functions $E, \phi_i, D_i$, representing encoder, filter, and decoder respectively, are each implemented as neural networks. To answer a given question, each decoder receives the part of the representation that is transmitted by its filter. The cost function is designed to minimise the error of the answer and the amount of information that is being transmitted from the encoder to decoders, which can be seen as minimising the amount of parameters that have to be communicated between agents.}
\label{fig:abstract_network_structure}
\end{figure*}
\begin{figure*}[ht!] 
\centering
\includegraphics[width=0.6\textwidth]{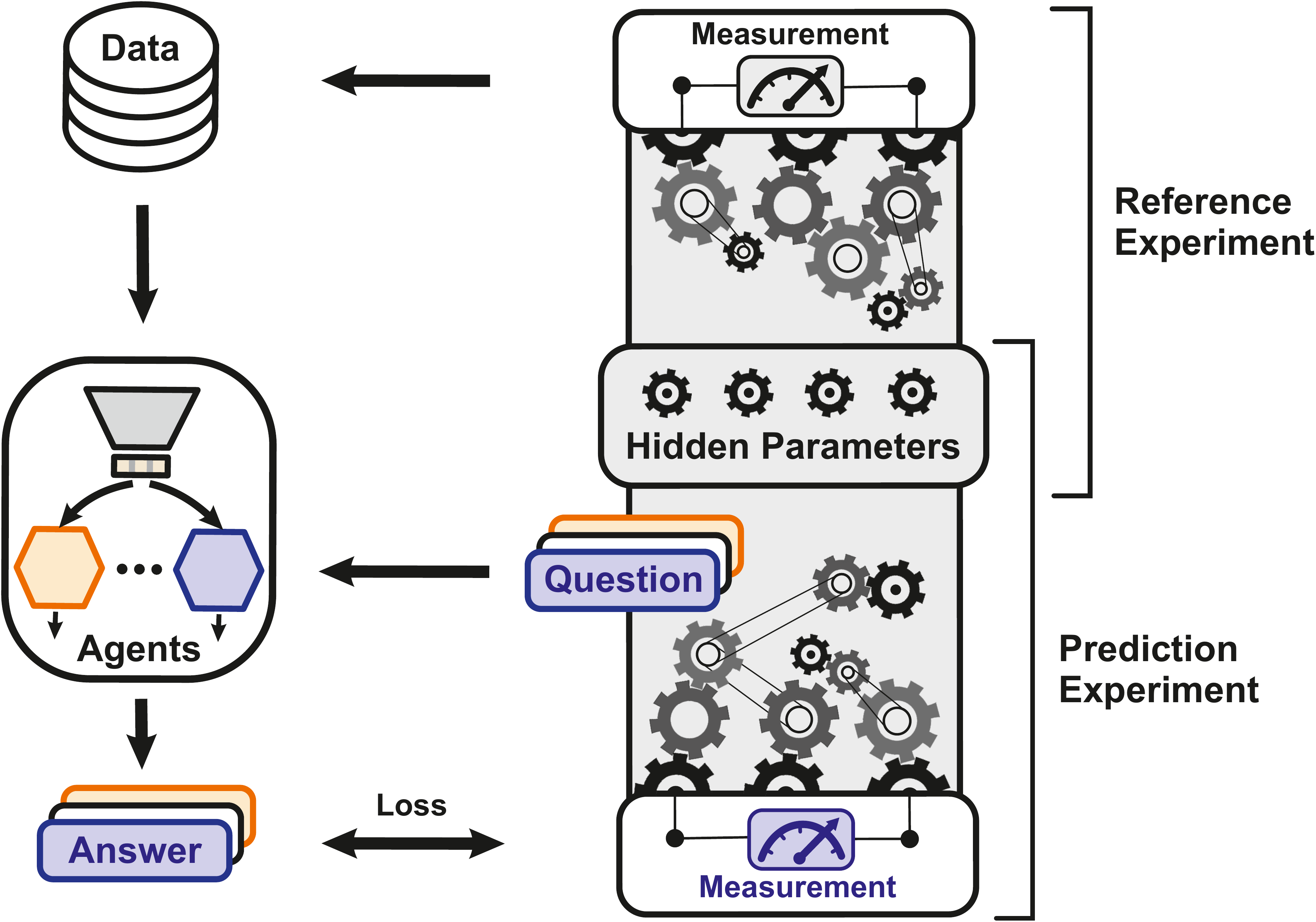} 
\caption{ 
\textbf{Experimental settings for representation learning.} Here, we illustrate the type of experiments that is used to provide data to agents. In our notion of experiments, experimental measurement results are governed by (hidden) parameters depicted as gears. When designing experiments (pictured as a complex network of gears) for the specific examples in Sec.~\ref{sec:experiments}, we start with \emph{reference} experiments whose measurement results (data) comprise a high-dimensional encoding of the hidden parameters. This data is referenced by other agents to make predictions, hence the name. With the encoding and additional information given as a question, our architecture can produce predictions (or answers) about the results of similar experimental settings, dubbed \emph{prediction} experiments.  
Insofar as different agents answer different questions, their prediction experiments may yield distinct outcomes. Those parts of an experiment that are associated with a (coloured) agent are depicted by the same colouring.
In our setting, questions may influence the prediction experiment in various ways. For instance, we consider experiments where questions are  encoded by specifying additional parameters of the experiment given to an agent as low- or high-dimensional representation (see Sec.~\ref{sec:experiments}). We also study the minimal case where questions are constant and just distinguish the tasks of different agents in Sec.~\ref{sec:reinforcement_learning}.
}
\label{fig:experiments_def}
\end{figure*}

Our setting is inspired by the idea that the physically meaningful parameters are those which are useful for answering different questions about or solving different tasks related to the same physical system. For instance, we use the parameters mass $m$ and charge $q$ to describe a particle because there exist operationally meaningful questions about the particle in which only the mass or only the charge is relevant, and other questions for which both are required. If we stored $m + q$ and $m - q$ instead (in some fixed units), we would still have the same information, but to answer a question just involving the mass, we would need both parameters instead of one; in contrast, there are few, if any, operationally meaningful questions for which only $m + q$ is relevant. Therefore, we say that $m$ and $q$ are operationally meaningful parameters, whereas $m + q$ and $m - q$ are not.
Note that we assume a specific notion of experiments in this paper and will continue to do so implicitly in the following. 
Here we understand an  \emph{experiment} as a stochastic function which maps a space of input parameters onto an output space representing measurement data such that the output distribution is reproducible for fixed parameters.
An \emph{experimental setting} is then an instance of an experiment with specified parameters. 
We assume that we can sample many different experimental settings, i.e., sample many instance of the same experiment with different parameters. 
For example, in an experiment involving a mass, we can sample many experimental settings with different (but possibly unknown) values for this mass.

\subsection{Communicating agents}

Here, we consider the following generic setting (see Fig.~\ref{fig:abstract_network_structure}(a)). 
Various agents have access to a physical system in form of e.g., measurement data. However, not all agents have access to the same data and some agents need to communicate with each other in order to answer a question or solve a task within this physical system. The constraint that agents need to communicate efficiently imposes structure on the representation of the communicated data. 
In the simplest case, a single \emph{encoding} agent \textit{A} makes an observation $o \in \mathcal{O}$ on a physical system with randomly chosen unknown parameters and generates a parameterised representation $r\in\mathcal{R}$, where $\mathcal{O}$ is the set of possible observations and $\mathcal{R}$ is a representational parameter space. 
For instance, the agent could observe a time series of particle positions and represent the velocity parameter. 
Other, \emph{decoding} agents $B_1, \dots, B_k$ are given questions $q_1, \dots, q_k$ randomly sampled from $\mathcal{Q}_1, \dots, \mathcal{Q}_k$, respectively, and are required to produce an answer $a_i(o,q_i)$.
For now, we assume that both the observation $o$ and the optimal answer $a^{*}_{i}(o, q_i) \in \mathcal{A}_i$ may be obtained directly from the respective experimental setting. 
In Sec.~\ref{sec:reinforcement_learning}, we consider the case where the optimal answer is not immediately apparent from an observation $o$ but may be learnt through reinforcement learning. Formally, one data sample consists of $\Big(o, \big(q_1, \dots, q_k\big), \big(a^{*}_{1}(o, q_1), \dots, a^{*}_{k}(o, q_k) \big) \Big)$. To generate the training data set, we collect such samples for many configurations of the unknown parameters of the physical system and many randomly chosen questions. 
By contrast, in Sec.~\ref{sec:reinforcement_learning} the training data is effectively generated by a trained reinforcement learning agent. 
In practice, we can represent observations, questions and answers as tuples of real numbers, and we will do so implicitly for the rest of this paper. 
Instead of having access to the entire observation $o$, $B_1, \dots, B_k$ only receive (part of) the encoding $r$. 
That is, $A$ is required to \textit{communicate} part of its representation to the other agents such that they can solve their respective tasks optimally.

The values $(o, q_i, a^{*}_i)$ are related to our notion of experiments (see Fig.~\ref{fig:experiments_def}) in the following way. An experiment $\mathcal{E} = (\mathcal{E}_1, \dots, \mathcal{E}_k)$ maps a parameter space $\Phi$ and a question space $\mathcal{Q}_i$ onto a result space $\mathcal{A}_i$, i.e. $\mathcal{E}_i:\Phi \times\mathcal{Q}_i \to \mathcal{A}_i$. 
The parameters in $\Phi$ may be (partially) hidden or not directly observable, and we would like to learn a representation for them. Therefore, we construct a reference experiment $\mathcal{E}_{r}$ which can be used to generate measurement data $o$ as a high-dimensional representation of the hidden parameters, i.e., $\mathcal{E}_{r}: \Phi \to \mathcal{O}$. This experiment is labeled the \emph{reference} because it provides the data which is used, or referenced, in parts by the other agents to make predictions.
Questions may be considered as additional parameters of the experiment of which we do not seek to find a representation, but which are useful for finding meaningful representations of the parameters in $\Phi$. Given the hidden parameters (encoded in $o\in\mathcal{O}$) and question (encoded in $q_i\in\mathcal{Q}_i$), the experiment produces some results $a^{*}_i\in\mathcal{A}_i$ which may be used to evaluate the answer given by an agent $B_i$.

\subsection{Learning objectives}
The operational criteria or learning objectives imposing structure on a representation take the form of different \emph{losses}, which are often referred to as \emph{priors} in representation learning~\cite{bengio_2017_consciousness,jonschkowski_2015_learning,francois_lavet_combined_2018}. In our case, the representation is generated under two criteria:
\begin{itemize}
\item With a \textit{prediction loss} we impose that agents need to learn to answer their questions as accurately as possible, given (part of) the representation. 
\item With a \textit{communication loss}, we impose that agents have to share the representation in the most data-efficient way.
\end{itemize}
In other words, the objective of the ensemble of agents $A, B_1, \dots, B_k$ is to correctly answer as many questions as possible, while also minimising the communication between $A$ and the other agents. Therefore, $A$ needs to disentangle its representation in a way that allows it to communicate the relevant parameters. More formally, we specify the encoding agent $A$ by a function $E: \mathcal{O} \to \mathcal{R} \equiv \R^l$ for some $l$ (see Fig~\ref{fig:abstract_network_structure}). 
This function can be thought of as an encoding from the high-dimensional experimental observation $o$ to a lower-dimensional vector of physically relevant parameters. 
In representation learning, the output of this function is called the representation.
Each decoding agent $B_i$ is specified by a \emph{filter} $\phi_i: \R^l \to \R^{l_i}$ and a function $D_i: R^{l_i} \times \mathcal{Q}_i \to \mathcal{A}_i$ such that the answer produced by the agent given an observation $o$ and question $q_i$ is $a_i(o, q_i) = D_i\Big(\phi_i\big(E(o)\big), q_i\Big)$ (see Fig~\ref{fig:abstract_network_structure}). 
The filter effectively restricts the agent's access to the representation by only transmitting a part of the representation; formally, $\phi_i(r_1, \dots, r_l) = (r_{j_1}, \dots, r_{j_{l_i}})$. 
We call $l_i$ the dimension $\dim(\phi_i)$ of the filter. Intuitively, one may imagine that the dimension and the indices $j_1, \dots, j_{l_i}$ of the transmitted components can be chosen by the agent. 
It is important that the filter is independent of the observation and question, since the transmission of parameters to agents should not depend on a particular data sample, but is instead viewed as a property of the \textit{theory} that applies to all data samples equally.  
The function $D_i$, called a decoder, takes the transmitted part of the representation and the question and produces an answer.
Ideally, agent $A$ produces a representation which allows each agent $B_i$ to answer its questions correctly while only accessing the smallest-possible part of the representation.

\subsection{Multiple encoding agents \label{sec:mult_enc}}
Up to now, we have assumed that there exists one agent \textit{A} who has access to the entire system to make an observation and to communicate its representation. However, just as different decoding agents $B_i$ only deal with a part of the system, we can consider the more general scenario of having multiple encoding agents $A_1, \dots, A_j$. In this scenario, each agent $A_i$ makes different measurements on the system. For example, one agent might make a collision experiment between two particles, while another observes the trajectory of a particle in an external field. Here, only the aggregate observations of all agents $A_1, \dots, A_j$ provide sufficient information about the system required for the agents $B_1, \dots, B_k$ to make predictions.

The formalisation is analogous to the previous section and we only sketch it here: we associate to each agent $A_i$ an encoder function $E_i$. The domain of the filter functions of the agents $B_1, \dots, B_k$ is now a cartesian product of the output spaces of the encoders (i.e., the output vectors of the encoders are concatenated and used as inputs to the filters). 

In the case where a physical system has an operationally natural division into $k$ interacting subsystems, a typical case would be to have the same number of encoding agents $A_1, \dots, A_k$ as decoding agents $B_1, \dots, B_k$, where both $A_i$ and $B_i$ act on the same $i$-th subsystem. Here, we expect that $A_i$ and $B_i$ are highly correlated, i.e., the filter for $B_i$ transmits almost all information from $A_i$, but less from other agents $A_j$. In this case, one can intuitively think of a single agent per subsystem $i$, that first makes an observation about that subsystem, then communicates with the other agents to account for the interaction between subsystems, and uses the information obtained from the communication to make a prediction about subsystem $i$.

\section{Model implementation}\label{sec:model}

Here, we discuss the details of the implementation and training of $E, \phi_i$ and $D_i$ (see Fig~\ref{fig:abstract_network_structure}). For brevity, we consider the case of a single encoding agent. The implementation of the multi-encoder scenario is analogous.
The functions $E, \phi_i, D_i$ are each implemented as neural networks. The encoder and decoder functions of the agents can be easily implemented using fully connected deep neural networks analogously to the architecture from Ref.~\cite{raban_2018_discovering}. To be more precise, the encoder is simply a deep neural network that maps a high-dimensional input to a low dimensional output consisting of a few so-called latent neurons. After being passed through the filter functions (which will be described in detail later) the representation is forwarded to all decoders. 
Additionally, each decoder receives a corresponding question vector as input. The decoder's neural network maps these to an output representing the answer.

While encoder and decoder are easy to implement, the implementation of filter functions $\phi_i$ poses a difficulty because these essentially need to learn a binary value, ``on'' or ``off'', for each of the latent neurons. Learning such discontinuous functions is not possible with standard backpropagation-based gradient descent. Therefore, we will need to introduce a smoothed version of this problem. However, we first need to understand the measure of success, i.e. the loss function that will be minimised.

\subsection{Learning objectives as loss functions}
As described above, the learning objective can be expressed in terms of loss or cost functions which are to be minimised by the ensemble of agents. The overall performance of the ensemble is quantified by a weighted sum of the following terms:
\begin{itemize}
\item Prediction losses $\mathcal{L}_{a, i} = ( a_i(o, q_i) - a^{*}_{i}(o, q_i))^2$ that measures how well the decoder answers the question.
\item A communication loss $\mathcal{L} _f = \sum_i \dim(\phi_i)$ that counts the total number of parameters transmitted to the agents $B_1, \dots, B_k$.
\end{itemize} 
In order to minimise the total cost, the neural network corresponding to the agent ensemble is then trained on a set of triples $\Big(o, \big(q_1, \dots, q_k\big), \big(a^{*}_{1}(o, q_1), \dots, a^{*}_{k}(o, q_k) \big) \Big)$. As described in the previous section, this data is provided in the form of measurement data obtained from various experimental settings. 

\subsection{Implementation of filters}

Due to the difficulty of implementing a binary value function with neural networks, we need to replace the ideal cost $\mathcal{L} _f$ by a comparable version with a smooth filter function.
To this end, instead of viewing the latent layer as the deterministic output of the encoder (the generalisation to multiple decoders is immediate), we consider each latent neuron $j$ as being sampled from a normal distribution $\mathcal{N}(\mu_j, \sigma_j)$. The sampling is performed using the renormalisation trick \cite{kingman_2013_auto}, which allows gradients to propagate through the sampling step. The encoder outputs the expectation values $\mu_j$ for all latent neurons. The logarithms of the standard deviations $\log(\sigma_j)$ are provided by neurons, which we call \emph{selection neurons}, that take no input and output a bias; the value of the bias can be modified during training using backpropagation. Using the logarithm of the standard deviation has the advantage that it can take any value, whereas the standard deviation itself is restricted to positive values. The ideal filter loss $\mathcal{L} _f = \sum_j \dim(\phi_j)$ is replaced by $\tilde{\mathcal{L} }_f = - \sum_j \log(\sigma_j)$.

The intuition for this scheme is as follows: when the network chooses $\sigma_j$ to be small (where the standard deviation of $\mu_j$ over the training set is used as normalisation), the decoder will usually obtain a sample that is close to the mean $\mu_j$; this corresponds to the filter transmitting this value. In contrast, for a large value of $\sigma_j$, a sample from $\mathcal{N}(\mu_j, \sigma_j)$ is usually far from the mean $\mu_j$; this corresponds to the filter blocking this value. The loss $\tilde{\mathcal{L} }_f$ is minimised when many of the $\sigma_j$ are large, i.e., when the filter blocks many values.

Instead of thinking of probability distributions, one can also view this scheme as adding noise to the latent variables, with $\sigma_j$ specifying the amount of noise added to the $j$-th latent neuron. If $\sigma_j$ is large, the noise effectively hides the value of this latent neuron, so the decoder cannot make use of it.

We also note that $\tilde{\mathcal{L} }_f$ is in principle unbounded. However, in practice this does not present a problem since the decoder can only approximately, but not perfectly, ignore the noisy latent neurons. For sufficiently large $\sigma_j$, the noise will therefore noticeably affect the decoders' predictions, and the additional loss incurred by worse predictions dominates the reduction in $\tilde{\mathcal{L} }_f$ obtained from larger values for $\sigma_j$. 

The success of this method to lead to an approximation of a binary filter depends on the weighting of the success loss in relation to the communication loss. This weight is a hyperparameter of the machine learning system.

\begin{figure*}[ht!]
  \centering
  \begin{subfigure}{0.4\textwidth}
  \centering
  \vspace{0.0cm}
  \includegraphics[width=\textwidth]{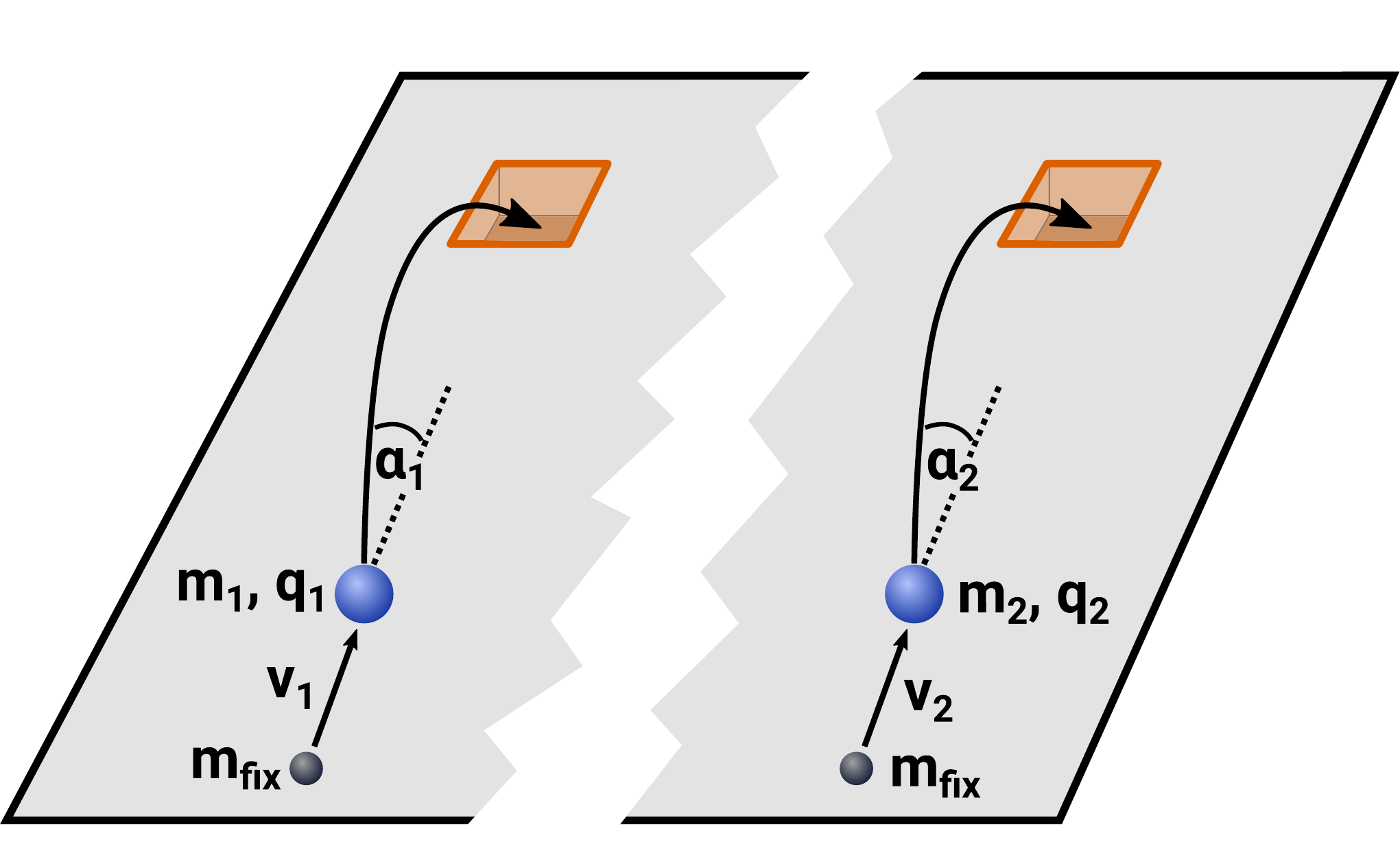} 
  \vspace{0.05cm}
  \caption{ ~ }  
\end{subfigure}
\quad\quad\quad
\begin{subfigure}{0.3\textwidth}
  \centering
  \includegraphics[width=\textwidth]{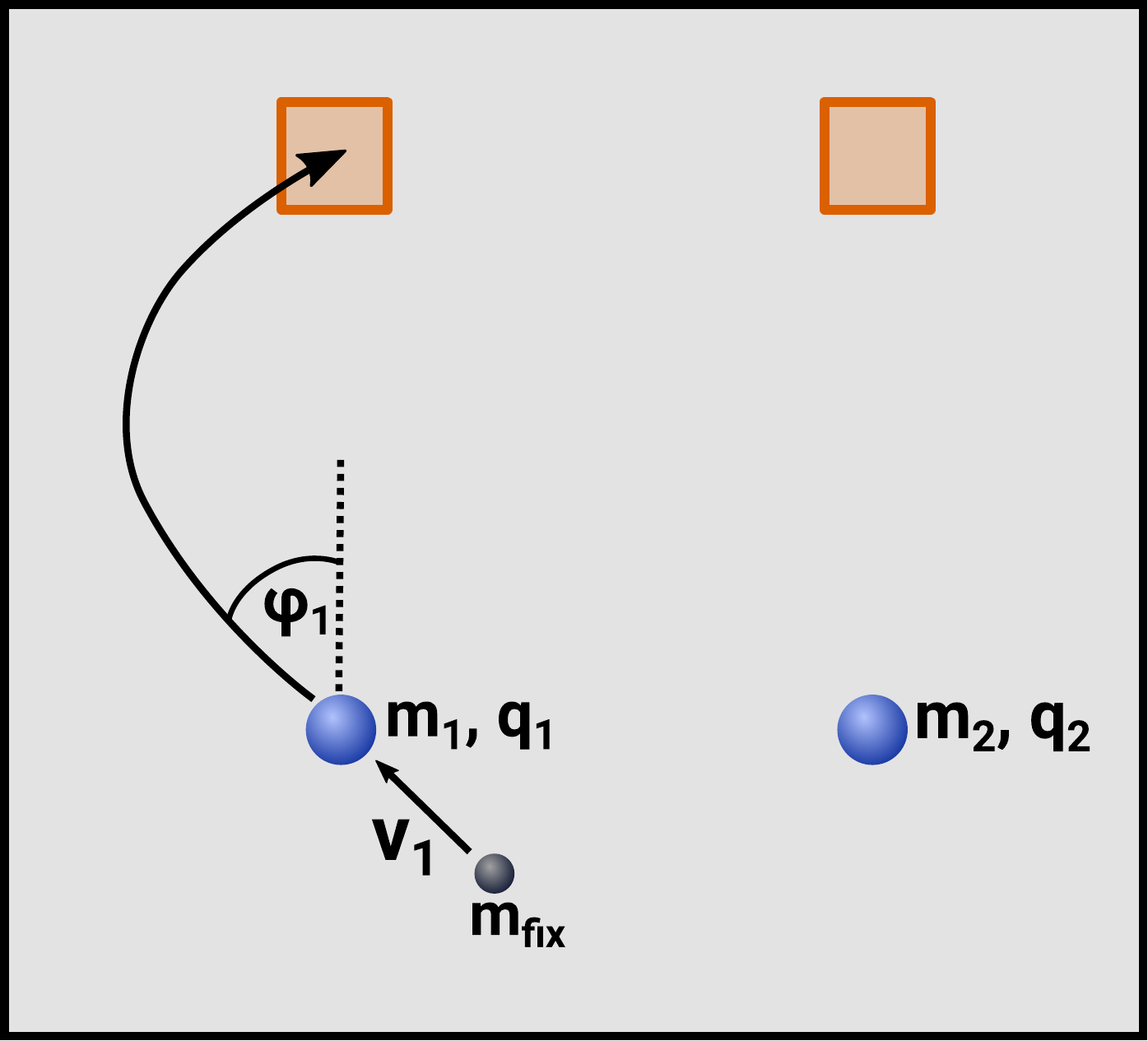}
  \caption{ ~ }
\end{subfigure}
\caption{\textbf{Toy example from classical mechanics with charged masses. }
(a) There are two separated decoding agents, so there is no interaction between their individual experimental setups. Each agent is required to \emph{shoot} a mass $m_i$ into a hole in the presence of a fixed gravitational field. They do this by elastically colliding a projectile of fixed mass $m_{\rm fix}$ with the mass $m_i$. The distance to the hole is fixed. The agents are given the velocity with which the projectile is shot as a question. The correct answer is the angle out of the plane of the table s.t. if the projectile is fired with the given velocity at this angle, the mass $m_i$ lands directly in the hole.
(b) Now we consider two decoding agents that are subject to the Coulomb interaction between their charged masses. Each agent is again required to shoot a projectile at its charged mass $(m_i, q_i)$. The charged mass will move in the Coulomb field of the other agent's charge. Similar to the first situation, each agent is given a velocity $v_i$ as a question and has to predict an angle (this time in the plane of the \emph{table}, i.e., gravity does not play a role) s.t. if the mass is fired at this angle, it will roll into the hole while the position of the other charge stays fixed. (The experiment is then repeated with the roles of the agents reversed, i.e., the agent that first fired his mass now fixes it at its starting position, and vice versa.)}
\label{fig:charged_setup}
\end{figure*}

\section{Examples with simple systems}\label{sec:experiments}

We demonstrate our method, both for single and multiple encoders, on two examples, one from classical mechanics, one from quantum mechanics. In all cases, the network finds a representation that complies with our operational requirements. 
We emphasise again that we refer to the term of experiment in order to describe a function mapping input parameters onto measurement data. An experimental setting is then an instance of an experiment with specified parameters and we assume access to a sampling method that produces experimental settings with varying parameters. In designing the following example experiments we follow the approach in Fig.~\ref{fig:experiments_def} specifying reference and prediction experiments.

\subsection{Charged masses} \label{sec:charged_masses}

\subsubsection{Setup}

We consider the setup shown in Fig.~\ref{fig:charged_setup}: take particles with masses $m_1, m_2$ and charges $q_1, q_2$, where both masses and charges are parameters that are varied between training examples. To generate the input data which is provided to the encoding agent $A$, we perform the following two \emph{reference experiments}:
\begin{enumerate}
\item We elastically collide each of the particles with masses $m_i$, initially at rest, with a reference mass $m_{\rm ref}$ moving at a fixed reference velocity $v_{\rm ref}$, and observe a time series of positions $(x_1, \dots, x_n)$ of the particle $m_i$ after the collision. In practice, we use $n=10$.
\item For each of the particles $(m_i, q_i)$, we place the particle at the origin at rest, and place a reference particle $(m_{\rm ref}, q_{\rm ref})$ with fixed mass and charge at a fixed distance $d_0$. Both particles are free to move. We  observe a time series of positions of the particle $(m_i, q_i)$ as it moves due to the Coulomb interaction between itself and the reference particle.
\end{enumerate}
Different agents now are required to answer different questions about the system in form of a prediction experiment (cf. Fig.~\ref{fig:experiments_def}). In this context, these questions can most easily be phrased as the agents trying to win games, both involving a target hole. The initial positions of the particles and the target holes are fixed.
\begin{itemize}
\item Agents $B_1$ and $B_2$ each are given projectiles with a fixed mass $m_{\rm fix}$. As question input, they are given the (variable) velocity $v_i$ with which this projectile will hit $m_i$. 
They can vary the angle $\alpha_i$ in the $yz$-plane with which they shoot this projectile against the mass $m_i$. After being hit, the mass will fly towards the target hole under the influence of gravity. The agent's goal is to hit the mass in precisely such a way that it lands directly in the hole, similar to a golfer attempting a lob shot that lands directly in the hole without bouncing. The prediction loss is given by the squared difference between the angle chosen by the agent and the \emph{correct} angle that would have landed the mass directly in the hole; this correct angle can be determined by experiments on the system. Alternatively, one could use the minimal distance of the trajectory of the particle to the hole as a cost function.
\item Similarly, agents $B_3$ and $B_4$ are given projectiles. The velocities of these projectiles are again given as a question input. The goal of the agent is to choose the angle $\phi_i$ in the $xy$-plane so that when the mass moves in the Coulomb field of the other mass (which stays fixed, then the experiment is repeated with the roles of moving and fixed mass reversed for the other agent), it will roll into the hole. 
\end{itemize}
In both cases, we restrict the velocities given as questions to ones where there actually exists a (unique) angle that makes the particle land in the hole.

\subsubsection{Results}

\begin{figure*}[ht!] 
\centering
\includegraphics[width=\textwidth]{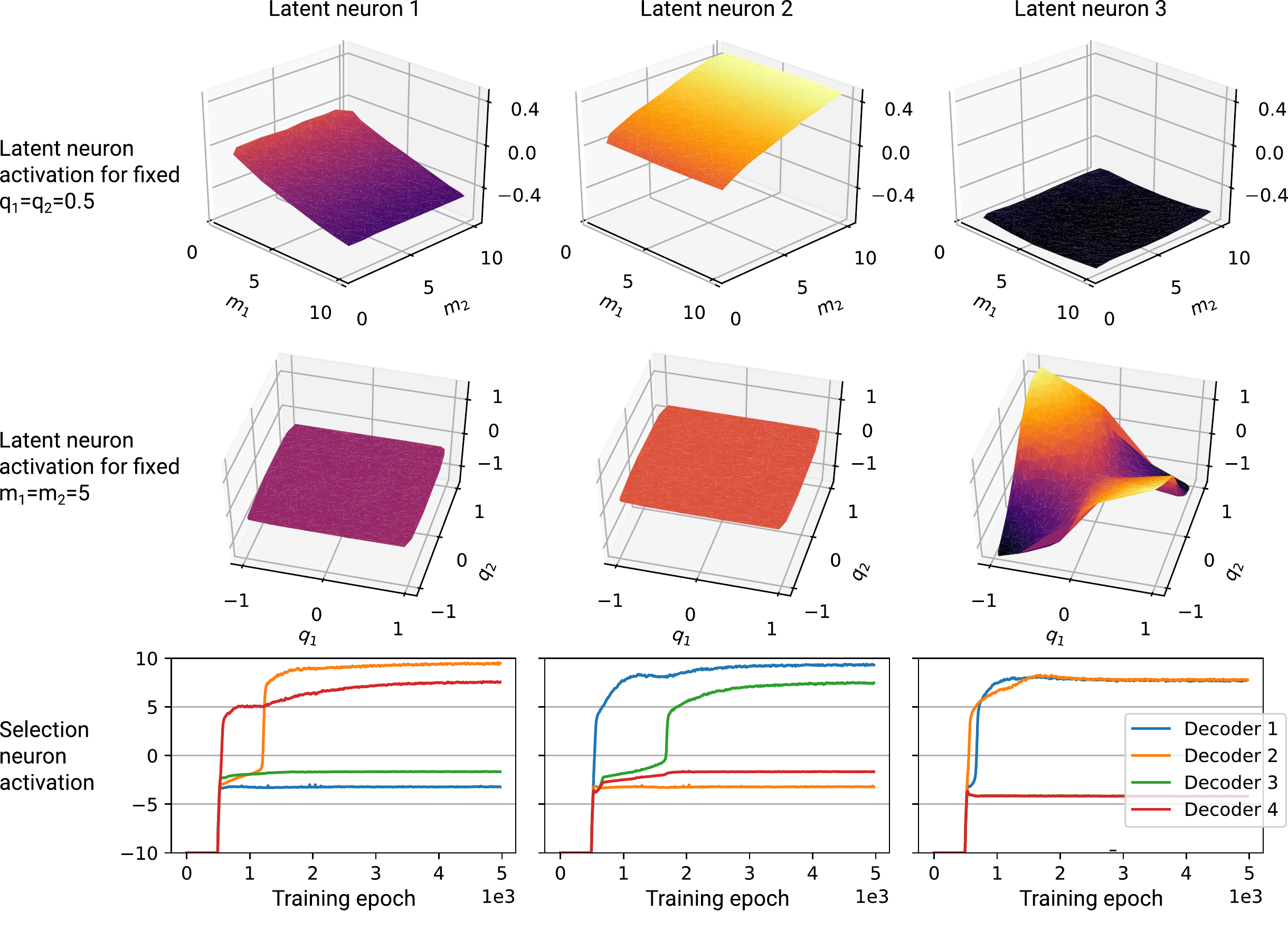}
\caption{\textbf{Results for the classical mechanics example with charged masses.} The used network has 3 latent neurons and each column of plots corresponds to one latent neuron. 
For the first row we generated input data with fixed charges $q_1=q_2=0.5$ and variable masses $m_1,m_2$ in order to plot the activation of latent neurons as a function of the masses. We observe that latent neuron 1 and 2 store the masses $m_1,m_2$ respectively while latent neuron 3 remains constant. In the second row, we plot the neurons' activation in response to $q_1,q_2$ with fixed masses $m_1, m_2=5$.  Here, the third latent neuron approximately stores  $q_1 \cdot q_2$, which is the relevant quantity for the Coulomb interaction while the other neurons are independent of the charges.
\\ The third row shows which decoder receives information from the respective latent neuron. Roughly, the $y$-axis quantifies how much information of the latent neuron is transmitted by the 4 filters to the associated decoder as a function of the training epoch. Positive values mean that the filter does not transmit any information. Decoders 1 and 2 perform non-interaction experiments with particles $(m_1,q_1)$ and $(m_2, q_2)$, respectively. Decoders 3 and 4 perform the corresponding interaction experiments. As expected, we observe that the information about $m_1$ (latent neuron 1) is received by decoders 1 and 3 and the information about $m_2$ (latent neuron 2)  is used by decoders 2 and 4.  Since decoders 3 and 4 answer questions about interaction experiments, the product of charges (latent neuron 3) is received only by them (the green line of decoder 3 in the last plot is hidden below the red one).
}   
\label{fig:charge_results}
\end{figure*}

To analyse the learnt representation, we plot the activation of the latent neurons for different examples with different (known) values of $m_1, m_2, q_1, q_2$ against those known values. This corresponds to comparing the learnt representation to a hypothesised representation that we might already have. The plots are shown in Fig.~\ref{fig:charge_results}. The first and second latent neurons are linear in $m_1$ and $m_2$, respectively, and independent of the charges; the third latent neuron has an activation that resembles the function $q_1 \cdot q_2$ and is independent of the masses. This means that the first and third latent neurons store the masses individually, as would be expected since the setup in Fig.~\ref{fig:charged_setup}(a) only requires individual masses and no charges. The third neuron roughly stores the product of the charges, i.e., the quantity relevant for the strength of the Coulomb interaction between the charges. This is used by the agents dealing with the setup in Fig.~\ref{fig:charged_setup}(b), where the particle's trajectory depends on the Coulomb interaction with the other particle.

\subsubsection{Multiple encoders}

One can easily adapt the above example to the multi-encoder setting described in Sec. \ref{sec:mult_enc}. Instead of having a single agent $A$, we use two agents $A_1$ and $A_2$, where agent $A_i$ only observes the results of the reference experiment associated with particle $i$. We provide detailed results in Appendix \ref{appendix:multi_enc}. The main finding is that there is no way for the encoding agents to directly encode the product of the charges $q_1 \cdot q_2$ anymore because each agent only has access to reference experiments involving a single charge. Instead, the representation produced by each encoding agent now stores $q_i$ individually (in addition to the mass $m_i$ as before). Hence, the additional structure imposed by splitting the encoding agent in two yields further disentanglement of the physical parameters of the system, allowing us to identify the individual charges rather than merely their product.

\subsection{Local representation of two-qubit states}

\subsubsection{Setup}

We consider a two-qubit system, i.e., a four dimensional quantum system. Finding a representation of such a system from measurement data is a non-trivial task called quantum state tomography~\cite{paris_quantum_2004}. In our operational setting, an agent $A$ has access to a reference experiment consisting of two devices, where the first device creates (many copies of) a quantum system in a state $\rho$, i.e., a positive semi-definite $4 \times 4$ matrix with unit trace, which depends on the parameters of the device. The second device can perform binary measurements (with output ``zero'' or ``one''), described by projections $\ket{\psi}\!\bra{\psi}$, where $\ket{\psi}$ is a pure state of two qubits.
\footnote{The probability to get outcome ``one'' for a measurement $\ket{\psi}\!\bra{\psi}$ is given by $p(\rho,\psi): =\bra{\psi} \rho  \ket{\psi}$.}
For the reference experiments, we fix $75$ randomly chosen binary measurements $\ket{\psi_1}\!\bra{\psi_1},\dots,\ket{\psi_{75}}\!\bra{\psi_{75}}$. 
For a given state $\rho$, the input to the encoder $A$ then consists of the probabilities to get ``one'' for each of the fixed $75$ measurements, respectively. 
The state $\rho$ is varied between training examples.

Three agents $B_1,B_2$ and $B_3$ are now required to answer different questions about prediction experiments with the two-qubit system:
\begin{itemize}
\item Agent $B_1$ and $B_2$ are asked questions about measurement output probabilities on the first and second qubit, respectively.
\item Agent $B_3$ is asked to predict joint measurement output probabilities on both qubits.
\end{itemize} 

More concretely, the question inputs consist of a binary measurement $\ket{\omega}\!\bra{\omega}$ (on one or two qubits, respectively), parametrised again by 75 randomly chosen projectors  $\ket{\phi_1}\!\bra{\phi_1},\dots,\ket{\phi_{75}}\!\bra{\phi_{75}}$. That is, the $i$-th question input corresponds to the probabilities $p(\omega,\phi_i)\coloneqq|\langle\phi_i|\omega\rangle|^2$ for all $i \in \{1,\dots,75 \}$.

\subsubsection{Results}
\begin{figure*}[ht!] 
\centering
\includegraphics[width=\textwidth]{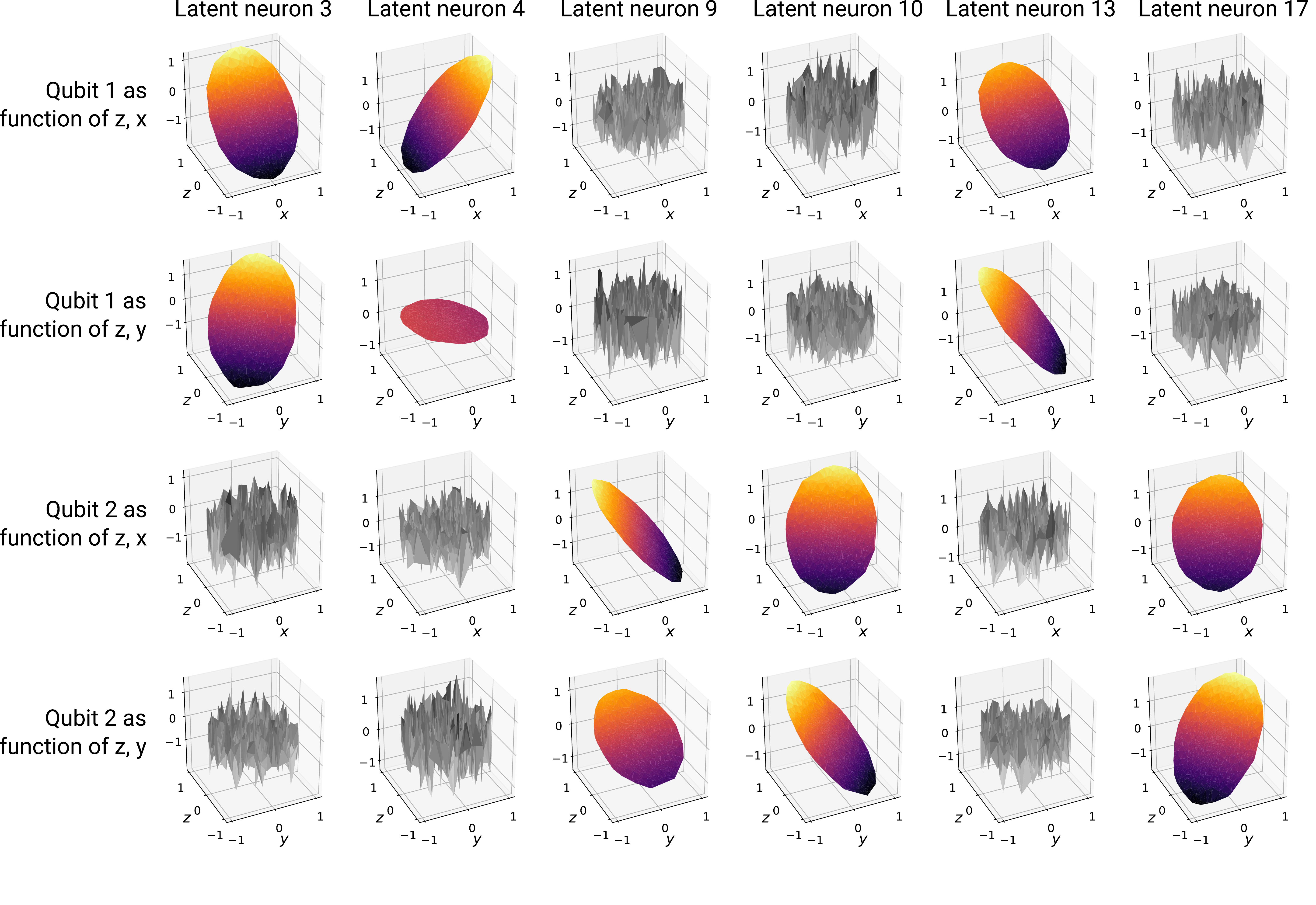}
\caption{\textbf{Results for the quantum mechanics example with two-qubit states.} We consider a quantum-mechanical system of two qubits. An encoder $A$ maps  tomographic data of a two-qubit state to a representation of the state. Three agents $B_1,B_2$ and $B_3$ are asked questions about the measurement output probabilities on the two-qubit system, where a question is given as the parameterisation of a measurement. Agents $B_1$ and $B_2$ are asked to predict measurement outcome probabilities on the first and second qubit, respectively. The third agent $B_3$ is tasked to predict measurement probabilities for arbitrary measurements on the full two-qubit system. Starting with 20 available latent neurons, we find that only 15 latent neurons are used to store the parameters required to answer the questions of all agents $B_1,B_2$ and $B_3$. Agent $B_3$ requires access to all parameters, while agents $B_1$ and $B_2$ need only access to two disjoint sets of three parameters, encoded in latent neurons 3,4,13 and 9,10,17, respectively. The plots show the activation values for these latent neurons in response to changes in the local degrees of freedom of each qubit, with the bottom axes of the plots denoting the components of the reduced one-qubit state $\rho=1/2(\id + x \, \sigma_x +y \, \sigma_y+z\,  \sigma_z)$ on either qubit 1 or 2.}  
\label{fig:qubit_xz}
\end{figure*}
We find that three latent neurons are used for each of the local qubit representations as required by agents $B_1$ and $B_2$. These local representations store combinations of the $x$-,$y$- and $z$-component of the Bloch sphere representation $\rho=1/2(\id + x \sigma_x +y \sigma_y+z \sigma_z)$ of a singe qubit (see Fig.~\ref{fig:qubit_xz}), where $\sigma_x,\sigma_y,\sigma_z$ denote the Pauli matrices. 
In general, a two-qubit mixed state $\rho$ is described by $15$ parameters, since a Hermitian $4 \times 4$ matrix is described by $16$ parameters, and one parameter is determined by the others due to the unit trace condition. 
Indeed, we find that the agent who has to predict the outcomes of the joint measurements accesses 15 latent neurons, including the ones storing the two local representations. 
Having chosen a network structure with 20 latent neurons, the 5 superfluous neurons are being successfully recognised and ignored by all of the agents $B_1,B_2$ and $B_3$. 
These numbers correspond to the numbers found in the analytical approach in Ref.~\cite{Gamel_2016}.
\newpage
\section{Reinforcement learning}\label{sec:reinforcement_learning}

So far, we have considered scenarios where agents make predictions about specific experimental settings and disentangle a latent representation by answering various questions. 
There, we understood \emph{answering} different questions as making \emph{predictions} about different aspects of a subsystem. 
Instead, we could have understood answers as \emph{sequences of actions} that achieve a specific goal. 
For example, such a (delayed) goal may arise when building experimental settings that bring about a specific phenomenon, or more generally when designing or controlling complex systems. 
In particular, we may view a prediction as a one-step sequence.

In the case of predictions, it is easy to evaluate the quality of a prediction, since we are predicting quantities whose actual value we can directly observe in Nature. 
In contrast, the correct sequences of actions may not be easily accessible from a given experimental setting: upon taking a first action, we do not yet know whether this was a good or bad action, i.e., whether it is part of a ``correct'' sequence of actions or not. 
Instead, we might only receive a few, sparsely distributed, discrete rewards while taking actions. In the typical case, there is only a binary reward at the end of a sequence of actions, specifying whether we reached the desired goal or not.
Even in a setting where a single action suffices to reach a goal, such a binary reward would prevent us from defining a useful answer loss in the same manner as before. To see this, consider the toy example in Fig.~\ref{fig:charged_setup}a again: the agent had to choose an angle $\alpha_i$, given a (representation of the) setting, specified by the parameters $(m_{fix}, m_1, q_1,m_2, q_2)$ and a question $v_i$, in order to shoot the particle into the hole. 
We assumed that we can evaluate the ``quality'' of the angle chosen by the agent by comparing it to the optimal angle (or equivalently measuring the distance between the agent's shot and the hole). 
If we instead only have access to a binary \emph{reward} specifying whether or not the agent successfully hit the (finite-sized) hole, we cannot define a smooth answer loss, which is required for training a neural network.

The problem that the feedback from the environment, i.e., the reward, is discrete or delayed can both be solved by viewing the situation as a reinforcement learning environment: given a representation of the setting (described by the masses and charges) and a question (a velocity), the agent can take different actions (corresponding to different angles at which the mass is shot) and receives a binary reward if the mass lands in the hole. Therefore, we can employ reinforcement learning techniques and learn the optimal answer. 

In reinforcement learning~\cite{sutton_1998_reinforcement}, an agent learns to choose actions that maximise its expected, cumulative, future, discounted reward. In the context of our toy example, we would expect a trained agent to always choose the optimal angle. 
Hence, predicting the behaviour of a trained agent would be equivalent to predicting the optimal answer and would impose the same structure on the parameterisation. 
In this example, the optimal solution consists of a single choice. In a more complex setting, it might not be possible to perform a (literal and metaphorical) hole-in-one.
 Generally, an optimal answer may require sequences of (discrete or continuous) actions, as it is for example the case for most control scenarios. In the settings we henceforth consider, questions might no longer be parameterised or given to the agent at all. That is, the question may be constant and just label the task that the agent has to solve.

In this section, we impose structure on the parameterisation of an experimental setting by assuming that different agents only require a subset of parameters to take a successful sequence of actions given their respective goals. 
To this end, we explain how experimental settings may be understood in terms of instances of a reinforcement learning environment and demonstrate that our architecture is able to generate an operationally meaningful representation of a modified standard reinforcement learning environment by predicting the behaviour of trained agents. 

Moreover, in Appendix~\ref{appendix:predicting_rl}, we lay out the details for the algorithm that allows us to generate and disentangle the parameterisation of a reinforcement learning environment given various reinforcement learning agents trained on different tasks within the same environment. There, we also prove that this algorithm produces agents which are at least as good as the trained agents while only observing part of the disentangled abstract representation.
The detailed architecture used for learning is described in Appendix~\ref{appendix:detail_RL} and is combing methods from GPU-accelerated actor-critic architectures~\cite{babaeizadeh_2017_reinforcement} and deep energy-based models~\cite{jerbi_2019_framework} for projective simulation~\cite{briegel_2012_projective}.

\subsection{Experiments as reinforcement learning environments}
\begin{figure*}[t!] 
\centering
\begin{subfigure}{0.45\textwidth}
\includegraphics[width=\textwidth]{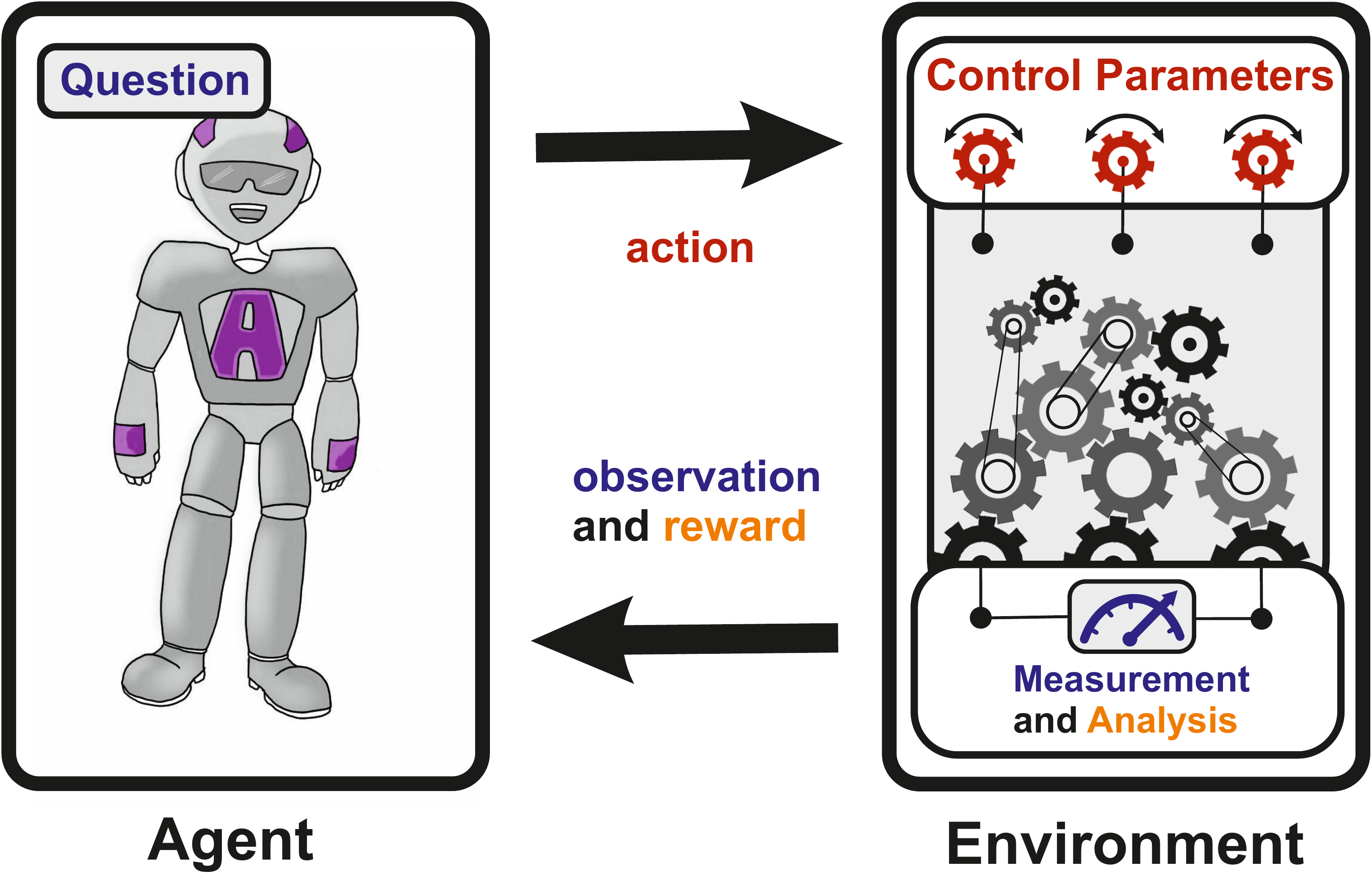} 
\caption{}
\end{subfigure}\quad\hspace{1cm}
\begin{subfigure}{0.35\textwidth}
\includegraphics[width=\textwidth]{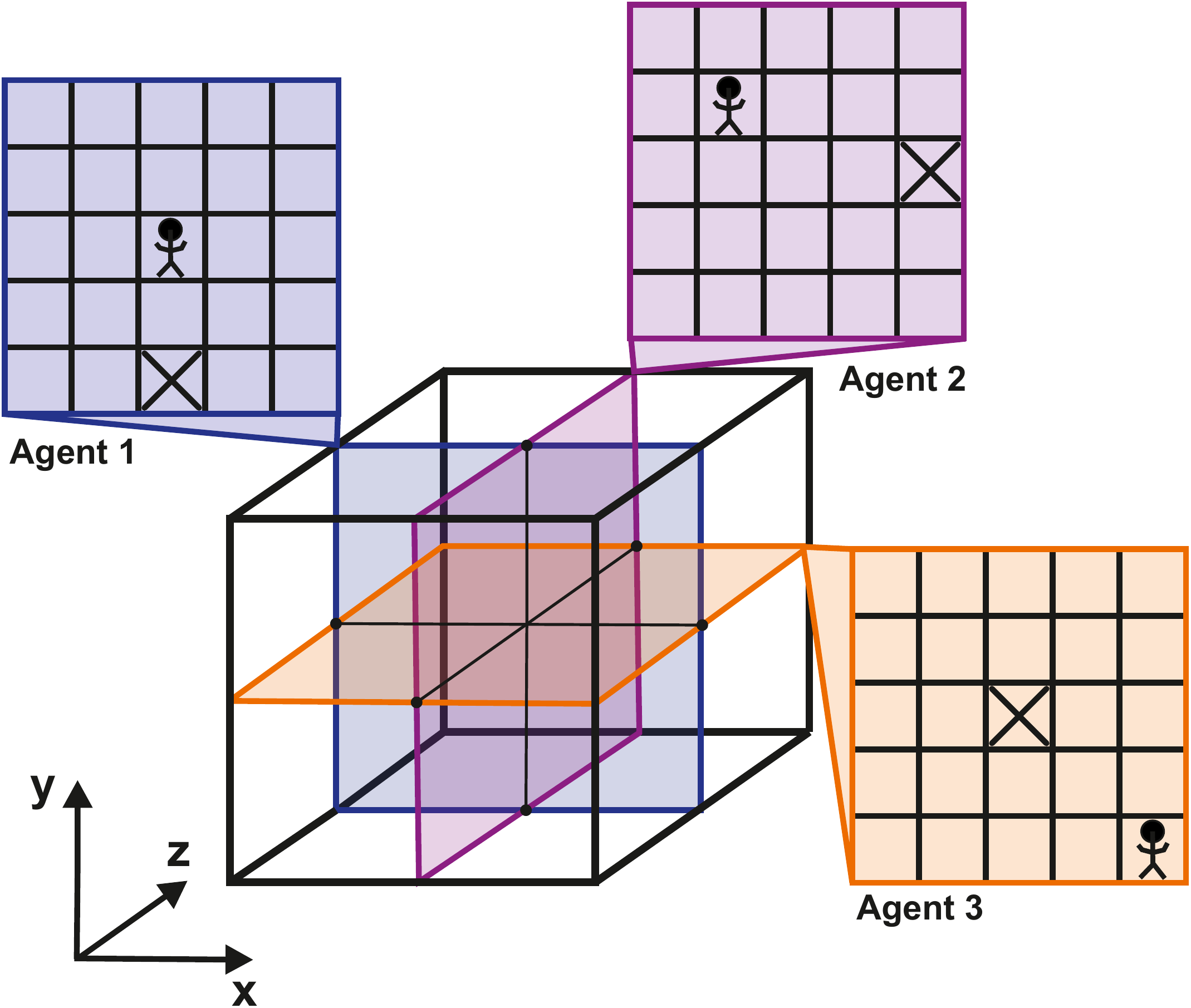} 
\caption{}
\end{subfigure}
\caption{\textbf{Experiments and reinforcement learning environments.} (a) In reinforcement learning an agent interacts with an environment. The agent can perform actions on the environment and receives perceptual information in form of an observation, i.e. the current state of the environment, and a reward which evaluates the agent's performance. An agent can also interact with an experimental setting (pictured as a complex network of gears) to answer its question by e.g., adjusting some control parameters (represented as red gears). It receives perceptual information in form of measurement data which may also have been analysed to provide an additional assessment of the current setting. (b) Sub-grid world environment. In this modified, standard reinforcement learning environment, agents are required to find a reward in a 3D grid world. Different agents are assigned different planes in which their respective rewards are located. Agents observe their position in the 3D gridworld and can move along any of the three spatial dimensions. An agent receives a reward once it has found the X in the grid. Then, the agent is reset to an arbitrary position and the reward is moved to a fixed position in a plane intersecting the agent's initial position.}  
\label{fig:env_to_exp}
\end{figure*}

In Ref.~\cite{melnikov_active_2018} the design of experimental settings has been framed in terms of reinforcement learning~\cite{sutton_1998_reinforcement} and here we formulate a similar setting: an agent interacts with experimental settings to achieve certain results. At each step the agent observes the current measurement data and/or setting and is asked to take an action regarding the current setting. This action may for instance affect the parameters of an experimental setting and hence might change the obtained measurement data. The measurement results are subsequently evaluated and the agent might receive a reward if the results are identified as ``successful''. 
The correspondence between experiments as described in this section and reinforcement learning environments can be understood as follows (cf. Fig.~\ref{fig:env_to_exp}a). An \textit{experimental setting} is interpreted as the current, internal state of an environment. The \textit{measurement data} then corresponds to the observation received from the environment. The agent performs an action according to the current observation and its question. Actions may affect the internal state of the experimental setting. For instance, the \textit{experimental parameters} describing the setting can be adjusted or chosen by an agent through actions. The reward function, which takes the current measurement data as input, describes the \textit{objective} that is to be achieved by an agent.

Since the same experiment can serve more than one purpose, we can have many agents interact with the same experimental setting to achieve different results. In fact, we can expect most experiments to be highly complex and have many applications. For instance, photonic experiments have a plethora of applications~\cite{erhard_2018_twisted} and various experimental and theoretical gadgets have been developed with these tools for different tasks~\cite{krenn_automated_2016,krenn_entanglement_2017,krenn_quantum_2017}. In this context, we may task various agents to develop gadgets for different task.
At first, we assume that all reinforcement learning agents have access to the entire measurement data. Once they have learnt to solve their respective tasks, we can employ our architecture from the previous section to predict each agent's behaviour. 
Effectively, we can then factorise the representation of the measurement data by imposing that only a minimal amount of information be required to predict the behaviour of each trained reinforcement learning agent. That is, we interpret the space of possible results in an experiment as high-dimensional manifold. 
When solving a given task however, an agent may only need to observe a submanifold which we want to parameterise.

Due to the close resemblance to reinforcement learning, we consider a standard problem in reinforcement learning  in the following and demonstrate that our architecture is able to generate an operationally meaningful representation of the environment. 
More formally, we consider partially-observable Markov decision processes~\cite{kaelbling_1998_planning} (POMDP). 
Given the stationary policy of a trained agent, we impose structure on the observation and action space of the POMDP by discarding observations and actions which are rarely encountered.  This structure defines the submanifold which we attempt to parameterise with our architecture. A detailed description of these environments is provided in Appendix~\ref{appendix:detail_RL_ENV}.

\subsection{Example with a standard reinforcement learning environment}

\subsubsection{Setup}

Here, we consider the simplest version of a task that is defined on a high-dimensional manifold while the behaviour of a trained agent may become restricted to a submanifold.
Consider a simple grid world task~\cite{sutton_1998_reinforcement} where all agents can move freely in a three-dimensional space whereas only a subspace is relevant to finding their respective rewards (see Fig.~\ref{fig:env_to_exp}b). 
Despite the apparent simplicity of this task, actual experimental settings may be understood as navigation tasks in complicated mazes~\cite{melnikov_active_2018}.
This reinforcement learning environment can be phrased as a simple game.
\begin{itemize}
\item Three reinforcement learning agents are positioned randomly within a discrete $12 \times 12 \times 12$ grid world. 
\item The rewards for the agents are located in a $(x,y)$-, $(y,z)$- and $(x,z)$-plane  relative to their respective initial positions. The locations of the rewards in their respective planes are fixed to $(6,11)$, $(11,6)$ and $(6,6)$.
\item The agents observe their position in the grid, but not the grid itself nor the reward. 
\item The agents can move freely along all three spatial dimension but cannot move outside the grid. 
\item An agent receives a reward if it can find the rewarded site within $400$ steps. Otherwise, it is reset to a random position and the reward is re-positioned appropriately in the corresponding plane.
\end{itemize}
Generally, in reinforcement learning the goal is to maximise the expected future reward. In this case, this requires an agent to minimise the number of steps until a reward is encountered.
Therefore, the optimal policy of an agent is to move on the shortest path towards the position of the reward within the assigned plane. Clearly, to predict the behaviour of an optimal agent, we require only knowledge of its position in the associated plane. We refer to Appendix~\ref{appendix:predicting_rl} for a concise protocol to \emph{predict behaviour} of a reinforcement learning agent. A detailed description of the architecture can be found in Appendix~\ref{appendix:detail_RL}. 

\subsubsection{Results}
\begin{figure*}[ht!] 
\centering
\includegraphics[width=1.\textwidth]{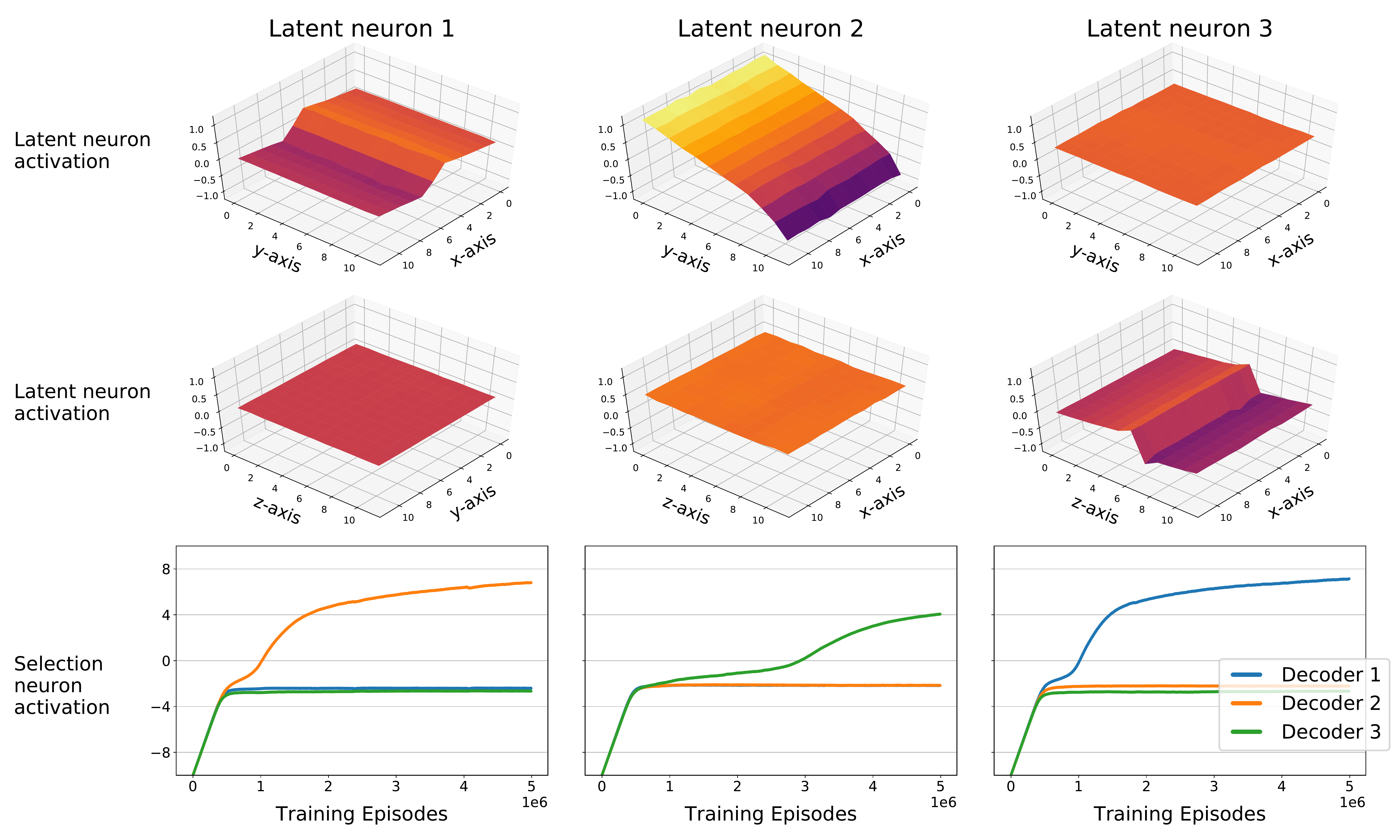}
\caption{\textbf{Results for the reinforcement learning example.} We consider a a $12 \times 12 \times 12$ 3D grid world.
The used network has 3 latent neurons and each column of plots corresponds to one latent neuron. 
For the first and second row we generated input data in which agent's position is varied along two axes and fixed to 6 in the remaining dimension.  The latent neuron activation is plotted as a function of the agent's position. We observe that the latent neurons 1,2 and 3 respond to changes in the $x$-, $y$- and $z$-position, respectively.
\\ The third row shows which decoder receives information from the each latent neuron. Roughly, the $y$-axis quantifies how much of the information in the latent neuron is transmitted by the 3 filters to the associated decoder as a function of the training episode. Positive values mean that the filter does not transmit any information. Decoder~1 has to make a prediction about the performance of a trained reinforcement learning agent whose goal is located within a $(x,y)$-plane relative to its starting position. We observe that decoder 1 indeed only receives information about the agent's $x$- and $y$-position, i.e. latent variables 1 and 2. Similarly, predictions made by decoders 2 and 3 only require knowledge of the agents' $(y,z)$- and $(x,z)$-position, respectively, which is confirmed by the selection neuron activations (the blue line of decoder 1 in the second plot is hidden behind the orange one). 
}
\label{fig:rl_results}
\end{figure*}

The optimal policy of an agent is to move on the shortest path towards the position of the reward within its assigned plane. Predicting the behaviour of an optimal agent, we require only knowledge of its position in the associated plane. 
Hence, the information about the coordinates should be separated such that the different agents have access to $(x,y), (y,z)$ and $(x,z)$, respectively. 
Using the minimal number of parameters, this is only possible if the encoding agent $A$ encodes the $x,y,z$ coordinates of the agents $B_1,B_2$ and $B_3$ and communicates their respective position in the plane\footnote{Because the observation space is discrete, an encoding agent can, in principle, ``cheat'' and encode multiple coordinates into a single neuron. In practice, this does not happen for sufficiently large state spaces.}.

We verify this by comparing the learnt representation to a hypothesised representation. For instance, we can test whether certain neurons respond to certain features in the experimental setting, i.e., reinforcement learning environment.
Indeed, it can be seen from Fig.~\ref{fig:rl_results} that the neurons of the latent layer only respond separately to changes in the $x,y$ or $z$ position of an agent respectively. Note that the encoding agent uses a nonlinear encoding of the $x$- and $z$-parameters. Interestingly, this reflects the symmetries in the problem: the reward is located at position $x=z=6$ whenever $x$ or $z$ are relevant coordinates for an agent, whereas for the $y$-coordinate, the reward is located at position $11$.  
The encoding used by the network in this example suggests that an encoding of discrete bounded parameters may carry additional information about the hidden reward function, which may eventually help to improve our understanding of the underlying theory.

\section{Conclusion}

Machine learning is rapidly developing into the newest tool in the physicists' toolbox~\cite{zdeborova_2017_new}. In this context, neural networks have become one of the most versatile and successful methods~\cite{lecun_deep_2015,silver_mastering_2016}. However, deep neural networks, while performing very well on a variety of tasks, often lack interpretability~\cite{olah_2018_the}. Therefore, representation learning, and in particular methods for learning interpretable representations, have recently received increased attention~\cite{Higgins2017, bengio_2017_consciousness,thomas_2018_disentangling,francois_lavet_combined_2018,jonschkowski_2015_learning,ried_2019_how}. In the scientific process in particular, representations of physical systems play a central role. 
To this end, we have developed a neural network architecture that can generate operationally meaningful representations within experimental settings. Roughly, we call a representation operationally meaningful if it can be shared efficiently between various agents that have different goals. We have demonstrated our methods for small toy examples in classical and quantum mechanics. Moreover, we have also considered cases where the experimental process may be framed as an interactive reinforcement learning scenario~\cite{melnikov_active_2018}. Our architecture also works in such a setting and generates representations which are physically meaningful and relatively easy to interpret.

In this work, we have interpreted the learnt representation by comparing it to some known or hypothesised representation. Instead, we could also seek to automate this process by employing unsupervised learning techniques that categorise experimental data by a metric defined by the response of different latent neurons. 
For the toy examples that we considered here, the learnt representation is small and simple enough to be interpretable by hand.
However, for more complex problems, additional methods for making the representation more interpretable may be required. 
For example, instead of using a single layer of latent neurons to store the parameters, recent work has shown the potential of semantically constrained graphs for this task~\cite{krenn_selfies_2018}. We expect that these methods can be integrated into our architecture to produce interpretable and meaningful representations even for highly complex latent spaces.

While we used an asynchronous, deep energy-based projective simulation model for reinforcement learning, our method for representation learning within reinforcement learning environments is independent of the exact reinforcement learning model and can be combined with other state-of-the-art techniques such as asynchronous, advantage actor-critic (A3C) methods~\cite{mnih_2016_asynchronous}. In fact, it may even be applied in settings with auxiliary tasks~\cite{jaderberg_2017_reinforcement} to develop meaningful representations. 

\section*{Source code and implementation details}
The source code, as well as details of the network structure and training process, including parameters, is available at \url{https://github.com/tonymetger/communicating_scinet} (for the first examples) and \url{https://github.com/HendrikPN/reinforced_scinet} (for the reinforcement learning part)
The networks were implemented using the Tensorflow~\cite{abadi_2015_tensorflow} and PyTorch~\cite{paszke_2017_automatic} library, respectively.

\section*{Contributions}
HPN, TM and RI contributed equally to the initial development of the project and composed the manuscript. 
HPN and TM performed the numerical work. 
SJ and LMT contributed to the theoretical and numerical development of the reinforcement learning part.
HJB and RR initialised and supervised the project.
All authors have discussed the results and contributed to the conceptual development of the project. 

\section*{Acknowledgments}
HPN, SJ, LMT and HJB acknowledge support from the Austrian Science Fund (FWF) through the DK-ALM: W1259-N27 and  SFB  BeyondC  F71. 
RI, HW and RR acknowledge support from from the Swiss National Science Foundation through SNSF project No. 200020\_165843 and through the National Centre of Competence in Research \textit{Quantum Science and Technology} (QSIT).
TM acknowledges support from ETH Z\"urich and the ETH Foundation through the \textit{Excellence Scholarship \& Opportunity Programme}, and from the IQIM, an NSF Physics Frontiers Center (NSF Grant PHY-1125565) with support of the Gordon and Betty Moore Foundation (GBMF-12500028).
SJ  also  acknowledges the Austrian Academy of Sciences as a recipient of the DOC Fellowship.  
HJB was also supported by the Ministerium f\"ur Wissenschaft, Forschung, und Kunst BadenW\"urttemberg  (AZ:33-7533.-30-10/41/1).
This work was supported by the Swiss National Supercomputing Centre (CSCS) under project ID da04.

\vspace{3em}
\appendix
\noindent{\LARGE{\sffamily {Appendix}}}

\section{Charged masses with multiple encoding agents}\label{appendix:multi_enc}

\begin{figure*}[ht!] 
\centering
\includegraphics[width=\textwidth]{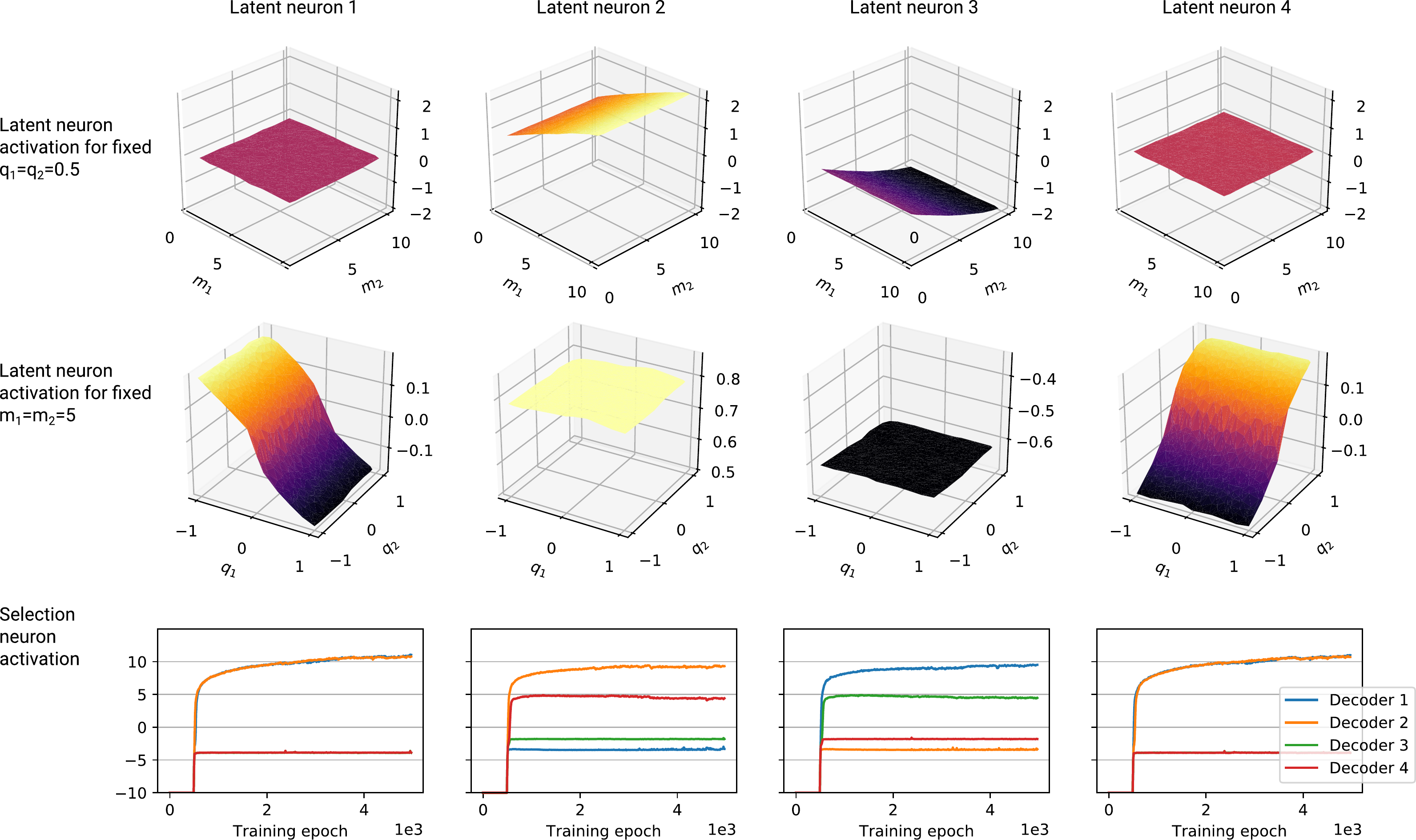}
\caption{\textbf{Results for the example with charged masses using two encoders.} The used network has 4 latent neurons and each column of plots corresponds to one latent neuron. For an explanation of how these plots are generated, see the caption of Fig. \ref{fig:charge_results}.
We observe that latent neurons 2 and 3 store the masses $m_1$ and $m_2$, respectively, while latent neurons 1 and 4 are independent of the mass.
Latent neurons 1 and 4 store (a monotonic function of) the charges $q_1$ and $q_2$, respectively, and are indepependent of $m_1$ and $m_2$. 
\\ The third row shows that the charges $q_1$ and $q_2$ are only transmitted to decoders 3 and 4, which are asked to make predictions about interaction experiments (the blue line of decoder 1 and the green line of decoder 3 are hidden under the orange and red lines, respectively, in both of these plots). The mass $m_1$, stored in the latent neuron 2, is transmitted to decoders 1 and 3, which are the two decoders that make predictions about particle 1. Analogously, $m_2$ is transmitted to decoders 2 and 4, which make predictions about particle 2.
}   
\label{fig:charge_results_multi_enc}
\end{figure*}

In this Section, we provide details about the representation learnt by a neural network with two encoders for the example involving charged masses introduced in Sec. \ref{sec:charged_masses}. The setup is the same as that in Section \ref{sec:charged_masses}, with the only difference being that we now use two encoders (the number of decoders and the predictions they are asked to make remain the same). Accordingly, we split the input into two parts: the measurement data from the reference experiments involving particle 1 are used as input for encoder 1, and the data for particle 2 are used as input for encoder 2. Each encoder has to produce a representation of its input. We stress that the two encoders are separated and have no access to any information about the input of the other encoder. The representations of the two encoders are then concatenated and treated like in the single-encoder setup; that is, for each decoder, a filter is applied to the concatenated representation and the filtered representation is used as input for the decoder.

The results for this case are shown in Fig. \ref{fig:charge_results_multi_enc}. Comparing this result with the single-encoder case in the main text, we observe that here, the charges $q_1$ and $q_2$ are stored individually in the latent representation, whereas the single encoder stored the product $q_1 \cdot q_2$. This is because, even though the decoders still only require the product $q_1 \cdot q_2$, no single encoder has sufficient information to output this product: the inputs of encoders 1 and 2 only contain information about the individual charges $q_1$ and $q_2$, respectively, but not their product. Hence, the additional structure imposed by splitting the input among two encoders yields a representation with more structure, i.e., with the two charges stored separately.

\section{Reinforcement learning environments for representation learning}\label{appendix:detail_RL_ENV}
\begin{figure*}[ht!] 
\centering
  \includegraphics[width=0.5\textwidth]{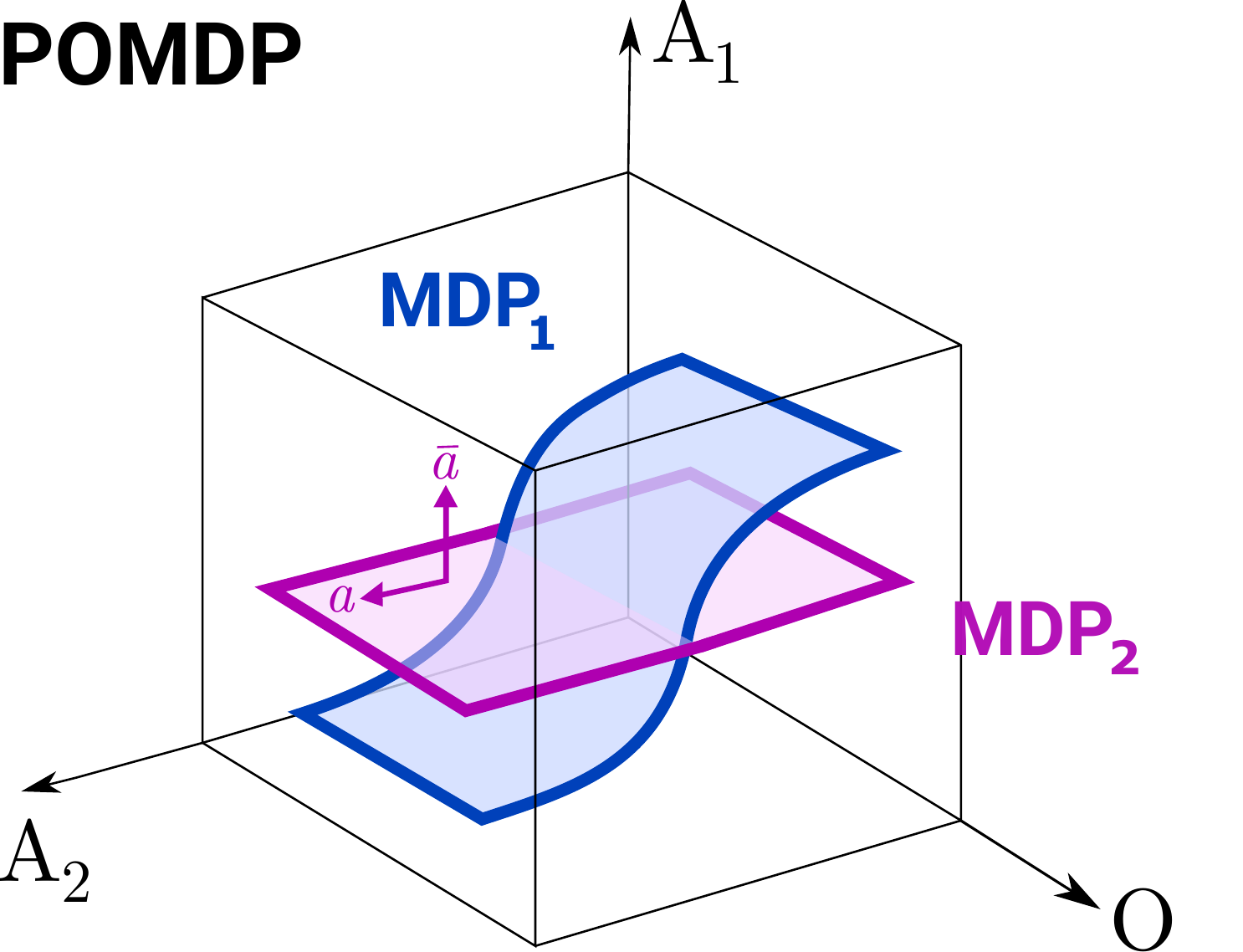} 
\caption{
\textbf{Observation and action space of the reinforcement learning environment. } 
The environment is described by a POMDP with an observation space $O$ and action space $A=A_1\times A_2$. The policy $\pi$ of the agent restricts the space $(O\times A)$ to a subset that we assume to describe an MDP. For example, $\mathrm{MDP}_1$ corresponds to a policy $\pi_1$ of one agent and $\mathrm{MDP}_2$ corresponds to a policy $\pi_2$ of another agent. 
The observation-action space is therefore restricted to a subset $(O \times A)|_{\pi_i (a|o)\geq 1/|A|}$ according to the learnt policy $\pi_i$ of an agent. Depicted is an action $a$ which is contained in this subset $\mathrm{MDP}_2$ together with an action $\bar{a}$ which is contained in the complement $(O \times A)|_{\pi_2 (a|o)< 1/|A|}$ of this subset.
}
\label{fig:statemani}
\end{figure*}
In this appendix, we give a formal description of the reinforcement learning environments that we consider for representation learning. As we will see, the sub-grid world example in the main text is a simple instance of such a class of environments. In general, we consider a reinforcement learning problem where the environment can be described as a Partially Observable Markov Decision Process~\cite{kaelbling_1998_planning} (POMDP), i.e., a MDP where not the full state of the environment is observed by the agent. We work with an observation space $O=\{o_1,...,o_{N}\}$, an action space $A=\{a_1,...,a_{M}\}$ and a discount factor $\gamma \in [0,1)$. This choice of environment does not reflect our specific choice of learning algorithm used to train the agent, as the latter does not construct so-called belief states that are commonly required to learn optimal policies in a POMDP. Rather, we want to show that our approach is applicable to slightly more general environments than Markov Decision Processes (MDPs) for which the learning algorithms we use are proven to converge to optimal policies in the limit of infinitely many interactions with the environment~\cite{sutton_1998_reinforcement,clausen_2019_on}. The generalisation to POMDPs still preserves the ``Markovianity'' of the environments and allows to consider only stationary (but not necessarily deterministic) policies $\pi(a|o)$, associated to stationary expected returns $R_\pi (o)$. 

Now consider an agent which exhibits some \emph{non-random} behaviour in this environment, which is characterised by a larger expected return than from a completely random policy.
Such a stationary policy may restrict observation-action space $(O \times A)$ to a subset $(O \times A)|_{\pi (a|o)\geq 1/|A|}$ of observations and actions likely to be experienced by the agent depending on its learnt policy $\pi$ and the environment dynamics. This notation indicates that, in any given observation, we discard actions that have probability less than random (i.e., less than $\frac{1}{|A|}$) of being taken by the agent, indicating that the agent's policy has learnt (un)favoring actions. In general, discarding actions also restricts the observation space. The subset $(O \times A)|_{\pi (a|o)\geq 1/|A|}$, along with the POMDP dynamics, describes a new environment. For simplicity, we assume that the restricted environment can be described by an MDP. This is trivially the case if the original environment is itself an MDP, and also the case for the sub-grid world environment discussed in the main text. The MDP inherits the discount factor $\gamma \in [0,1)$ of the original POMDP, which allows us to consider w.l.o.g.\ finite-horizon MDPs\footnote{An infinite-horizon MDP with discount factor $\gamma \in [0,1)$ can be $\epsilon$-approximated by a finite-horizon MDP with horizon $l_\text{max} = \log_\gamma (\frac{\epsilon(1-\gamma)}{\max_o |R(o)|})$.}, which are MDPs of finite episodes lengths (here, we set the maximum length to $3l_\text{max}$). A conceptual view on this POMDP restricted by policies is provided in Fig.~\ref{fig:statemani}.

\section{Representation learning in reinforcement learning environments}\label{appendix:predicting_rl}

In our approach to factorising abstract representations of reinforcement learning agents, we assume that an agent's policy can impose structure on an environment (as described in Appendix~\ref{appendix:detail_RL_ENV}) and we want this structure to be reflected in its latent representation. Therefore, decoders need to predict the \emph{behaviour} of a reinforcement learning agent while requiring minimal knowledge of the latent representation. However, we still lack a definition of what it means for a decoder to predict the behaviour of an agent. Here, we consider decoders predicting the expected rewards for these agents given the representation communicated by the encoder. Later, we show that this is enough to produce a policy which is at least as good as the policy of the reinforcement learning agent. 

To be precise, each decoder attempts to learn the expected return $R_{\pi}(o,a)$ given an observation-action pair $(o,a)\in (O \times A)|_{\pi (a|o)\geq 1/|A|}$ under the policy $\pi$ of an agent. For observation-action pairs outside the restricted subset we assign values $0$. The input space of the decoder and the restriction to the subset is illustrated in Fig.~\ref{fig:statemani}.  In fact, decoders not only learn to predict $R$ for a single action but for a sequence of actions $\{a^{(1)},\dots, a^{(l)}\}_l$ with length $l\geq1$.  This is because it can help stabilise the latent representation of environments with small actions spaces and simple reward functions. In practice however, $l=1$ is sufficient to obtain a proper representation. In the same way, we can help  to stabilise the latent representation by forcing an additional decoder to reconstruct the input from the latent representation. For brevity, we write $\{a^{(i)}\}_l$ for sequences of actions of length $l$.

The method described in this appendix, allows us to pick a number of reinforcement learning agents that have learnt to solve various problems on a specific kind of reinforcement learning environment (see Appendix~\ref{appendix:detail_RL_ENV}) and parameterise the subspaces relevant for solving their respective tasks. Specifically, the procedure splits into three parts:
\begin{enumerate}
\item[(i)]{Train reinforcement learning agents.}
\item[(ii)]{Generate training data for representation learning from reinforcement learning agents (see Appendix~\ref{appendix:policy_data}).}
\item[(iii)]{Train encoders with decoders on training data such that they can reproduce (w.r.t. performance) the policy of the reinforcement learning agents (see Appendix~\ref{appendix:proof_policy_data}).}
\end{enumerate}

The purpose of this Appendix is to prove that the trained decoders contain enough information to derive policies that perform as well as the ones learnt by their associated agents. Only if this is the case, we can claim that the structure imposed by the decoder reflects the structure imposed on the environment by an agent's policy.  To that end, we start by (ii) introducing the method to generate the training data, followed by (iii) a construction of a policy from a trained decoder with given performance bounds.\\

\subsection{Training data generation}\label{appendix:policy_data}
\begin{figure*}[ht!]
	\centering
	\includegraphics[width=0.8\textwidth]{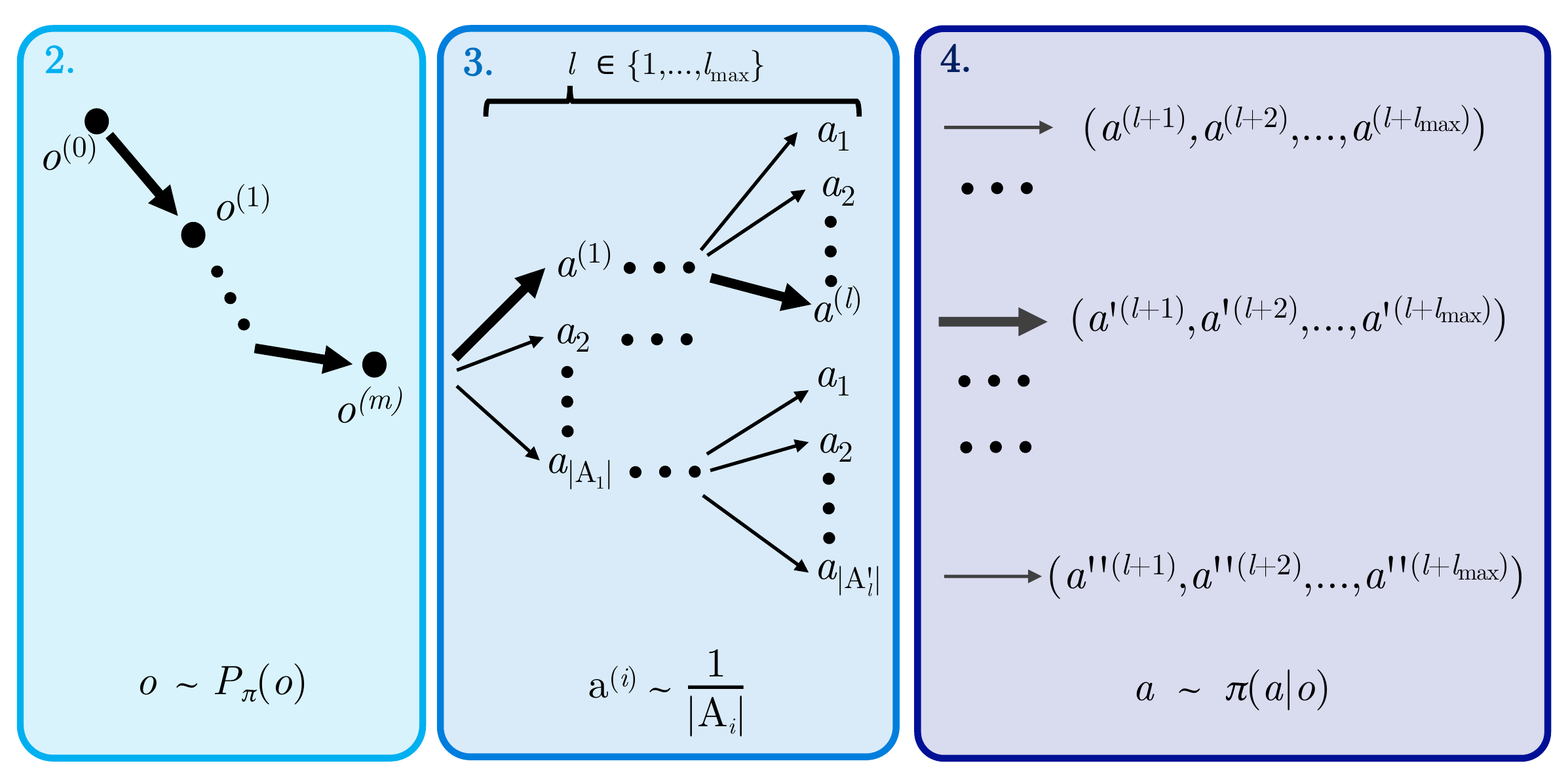}
	\caption{\textbf{Generating training data. }The training data of the decoder is sampled from the environment in three main steps (for a complete list see Appendix~\ref{appendix:policy_data}). 
2. Starting from an initial observation $o^{(0)}$, the observation $o^{(m)}$ is reached by following the policy $\pi$ of the agent. This is equivalent to sampling $o^{(m)}$ from the probability distribution $P_{\pi}(o^{(m)})$. 
3. A sequence of $l$ actions $\{a^{(i)}\}_l$ is randomly sampled from the action space, as restricted by the subset $(O \times A)|_{\pi (a|o)\geq 1/|A|}$, and executed on the environment.  We write $a_i$ $\forall i=1,\dots,|A_j|$ for actions restricted to the subset at a given observation $o^{m+j}$.
4. Finally, $l_\text{max}$ actions are drawn from the policy $\pi$ with the restriction $\pi (a|o)\geq 1/|A|$. These last actions are executed on the environment and their associated rewards are collected to compute an estimate of $R_{\pi}(o,\{a^{(i)}\}_l)$.}
	\label{fig:sequence}
\end{figure*}
The decoders are trained to predict the return values $R_{\pi}(o,\{a^{(i)}\}_l)$ for observations $o$ and sequences of actions $\{a^{(i)}\}_l$ of arbitrary length $l \leq l_\text{max}$, given a policy $\pi$. 
The training data is then generated as follows
(see Figure \ref{fig:sequence}): 
\begin{enumerate}
\item Sample two numbers $m,l$ uniformly at random from $\{1, \ldots, l_\text{max}\}$.
\item Start an environment rollout with the trained agent's policy $\pi$ for $m$ steps until the observation $o^{(m)}$ is reached.
\item Continue the rollout with $l$ actions which are sampled uniformly at random from the action space as restricted by the subset $(O \times A)|_{\pi (a|o)\geq 1/|A|}$\footnote{Note that these actions need to be sampled sequentially from the current policy of the agent, given an observation.}.
\item The rollout is completed with $l_\text{max}$ steps according to the policy $\pi$ of the agent restricted to the subset.
\item The rewards $r_j$ associated to the last $l_\text{max}$ steps are collected and used to evaluate an estimate of $R_{\pi}(o,\{a^{(i)}\}_l) = \sum_{j=1}^{l_\text{max}} \gamma^{j-1} r_j$.
\item Collect a tuple consisting of observation $o^{(m)}$,  actions $\{a^{(i)}\}_l$ and reward $R_{\pi}(o,\{a^{(i)}\}_l) \big)$.
\item Collect tuples $\big(o^{(m)},  \{\bar{a}^{(i)}\}_l,   0 \big)$ for all actions $\bar{a}^{(i)}$ which are not in the restricted subset $(O \times A)|_{\pi (a|o)\geq 1/|A|}$.
\item Repeat the procedure.
\end{enumerate}
Note, that this algorithm does not require any additional control over the environment beyond initialisation and performing actions. That is, it can be generated \emph{on-line} while interacting with the environment. In the case of a deterministic MDP and policy, one iteration of this algorithm yields the exact values of $R_{\pi}(o,\{a^{(i)}\}_l)$. In the case of a stochastic MDP or policy, one obtains instead an unbiased estimate of these values due to the possible fluctuations caused by the stochasticity of the environment dynamics and the policy. Repeated iterations of the algorithm followed by averaging of the estimates allows to decrease the estimation error. We neglect this estimation error in the next Section.

The collected tuples are used  to train the encoder and decoder through the answer loss $\mathcal{L}_a$ as discussed in the main text. In practice, short action sequences are sufficient to factorise the abstract representation of the trained agents. In the example of the main text, $l=1$ was used. We kept the general description of the return function with arbitrary sequence lengths as a possible extension for more stable factorisations.

\subsection{Reinforcement learning policy from trained decoders}\label{appendix:proof_policy_data}
Let us call $R_\text{NN}$ the function learnt by the decoder. We prove that a policy $\pi'$ satisfying $R_{\pi'}(o^{(0)})\geq R_{\pi}(o^{(0)})$ $\forall o^{(0)}$ in the MDP can be constructed from the decoder if it was trained with a certain loss $\epsilon$.
\begin{theorem}
Given a POMDP with observation-action space $O \times A$ and a policy $\pi$ that restricts the POMDP into an MDP with observation-action space $(O \times A)|_{\pi (a|o)\geq 1/|A|}$, there exists a policy $\pi'$ that satisfies $R_{\pi'}(o^{(0)})\geq R_{\pi}(o^{(0)})$ $\forall o^{(0)}$ in the MDP and that can be derived from a function which is $\epsilon$-close (in terms of a mean squared error), with $\epsilon>0$, to:
\begin{equation*}
\widetilde{R}_\pi (o,a) =
\begin{cases}
R_{\pi}(o,a) & \textrm{if }(o,a) \in (O \times A)|_{\pi (a|o)\geq 1/|A|}\\
\quad0 & \textrm{otherwise}
\end{cases}
\end{equation*}
\end{theorem}
\begin{proof}
For clarity, we first prove that the construction of $\pi'$ is possible if the return values are learnt perfectly, i.e., the training loss $\mathcal{L}$ is zero.
Later, we relax this assumption and show that the proof still holds for non-zero values of the loss.

We choose the loss function to be a weighted mean square error on the subset extended to arbitrary length action sequences, i.e., $(O\times \bigcup_{k=1,\ldots, l_\text{max}}A^{k})|_{\pi (a|o)\geq 1/|A|}$,
\begin{align*}
\mathcal{L} &=\sum_{o,\{a^{(i)}\}_l} P_{\pi}(o)\frac{1}{l_\text{max}\prod_i |A_i|}\\
&\times (R_{\pi}(o,\{a^{(i)}\}_l)-R_\text{NN}(o,\{a^{(i)}\}_l))^2.
\end{align*} 
An analogous approach yields similar results for other loss	 functions.
Here, $P_{\pi}(o)$ is the probability that the observation $o$ is obtained given that the agent follows the policy $\pi$ and $A_i$ is the action space from which the action $a^{(i)}$ is sampled, as restricted by the subset.
Now, let us further restrict the sum to action sequences of length one, i.e.,
\begin{equation*}
\mathcal{L'}=\sum_{o,a} P_{\pi}(o)\frac{1}{l_\text{max}|A_1|}(R_{\pi}(o,a)-R_\text{NN}(o,a))^2,
\end{equation*}
for which it is easily verified that $\mathcal{L}' \leq \mathcal{L}$.

Using $R_\text{NN}$, we derive the following policy: 
\begin{equation}\label{eq:decoder-policy}
\pi'(a|o)=
\begin{cases}
1\quad\textrm{if }a=\textrm{argmax}_{a'} R_\text{NN}(o,a')\\
0\quad\textrm{otherwise}\\
\end{cases}
\end{equation}
Since $R_\text{NN}(o,a)$ corresponds to the return of the policy $\pi$ after observing $o$ and taking action $a$, maximising this return hence leads to a return $R_{\pi'}(o)\geq R_{\pi}(o)\ \forall o \in O_\text{MDP}$.\\

In the following, we discuss the implications of the decoder not learning to reproduce $R_{\pi}$ perfectly, i.e., $\mathcal{L} = \epsilon > 0$. More precisely, we derive a bound on $\epsilon$ under which a policy $\pi'$ satisfying $R_{\pi'}(o^{(0)})\geq R_{\pi}(o^{(0)})$ $\forall o^{(0)}$ in the MDP can still be constructed from the decoder.

The decoder can be used to construct the policy $\pi'$ defined in Eq.\ (\ref{eq:decoder-policy}) if the approximation error of $R_\text{NN}$ is small enough to distinguish the largest and second-largest return values $R_\pi(o,a)$ given an observation $o$. In the worst case, this difference can be as small as the smallest difference between any two returns given an observation
\begin{equation*}
\epsilon' = \gamma^{l_\text{max}}\delta_R,
\end{equation*} 
where $\delta_R = \min_i{|r_{i+1}-r_{i}|}$ is the minimal non-zero difference between any two values the reward function of the environment can assign (including a reward $r=0$).\\
Let us set,
\begin{equation*}
\mathcal{L'} \leq \epsilon = \frac{\gamma^{2l_\text{max}} \delta_R^2 \delta_{\pi}}{16|A|l_\text{max}}
\end{equation*}
where $\delta_{\pi}=\min_{o \in O_\text{MDP}}\{P_{\pi}(o)\ |\ P_{\pi}(o)\neq 0\}$. That is,
\begin{equation*}
\sum_{o,a} P_{\pi}(o)\frac{1}{l_\text{max}|A_1|}(R_{\pi}(o,a)-R_\text{NN}(o,a))^2 \leq \frac{\gamma^{2l_\text{max}} \delta_R^2 \delta_{\pi}}{16|A|l_\text{max}}
\end{equation*}
and hence, $\forall (o,a) \in (O \times A)|_{\pi (a|o)\geq 1/|A|}$
\begin{align*}
P_{\pi}(o)\frac{1}{l_\text{max}|A_1|}(R_{\pi}(o,a)-R_\text{NN}(o,a))^2 &\leq \frac{\gamma^{2l_\text{max}} \delta_R^2 \delta_{\pi}}{16|A|l_\text{max}}\\
(R_{\pi}(o,a)-R_\text{NN}(o,a))^2 &\leq \frac{\gamma^{2l_\text{max}} \delta_R^2}{16}\\
|R_{\pi}(o,a)-R_\text{NN}(o,a)| &\leq \frac{\epsilon'}{4}.
\end{align*}

It is sufficient for $R_\text{NN}$ to approximate $R_\pi$ with precision $\frac{\epsilon'}{4}$.
Therefore, it is sufficient to bound the error of the loss function $\mathcal{L}$ by 
\begin{equation*}
\epsilon \leq \frac{\gamma^{2l_\text{max}} \delta_R^2 \delta_{\pi}}{16|A|l_\text{max}}. 
\end{equation*}
\end{proof}

This worst case analysis shows that the error needs to be exponentially small with respect to the parameters of the problem so that we can derive strong performance bounds of the policy on the entire subset. In practice, we expect to be able to derive a functional policy even with higher losses during the training of the decoder.

\section{Model implementation for representation learning in reinforcement learning environments}\label{appendix:detail_RL}
\begin{figure*}[ht!] 
\centering
\includegraphics[width=0.9\textwidth]{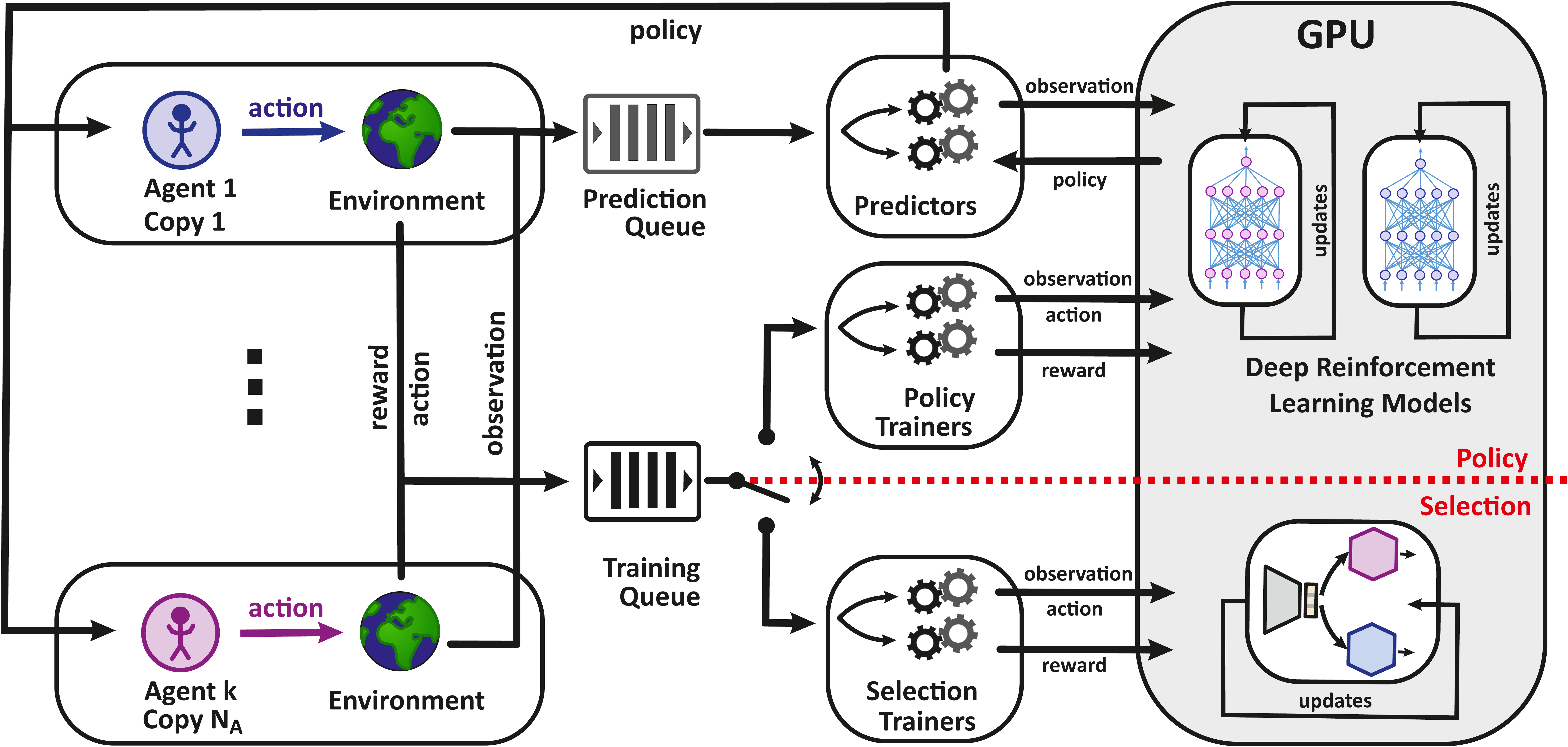} 
\caption{\textbf{Architecture for representation learning in reinforcement learning settings.}
We store neural network models in the shared memory of a graphics processing unit (GPU). As in Ref.~\cite{babaeizadeh_2017_reinforcement}, we make use of an asynchronous approach to reinforcement learning. 
$N_A$ copies for each of $k$ different agents interact with copies of the environment. Observations are queued and transferred to the GPU by prediction processes which also distribute policies, returned by the GPU, to the agents. Batches of observations, actions and rewards are queued and transferred to the GPU by trainer processes for updating the neural networks. On the GPU, batches of observations obtained from predictors are evaluated with deep energy-based projective simulation models~\cite{jerbi_2019_framework} to obtain a policy, and batches from policy trainers are used to update the model via the loss in Eq.\ (\ref{eq:training_dps}).
Everything above the red dotted line concerns the training of the reinforcement learning agents' policy analogous to Ref.~\cite{babaeizadeh_2017_reinforcement}. Below the dotted line, we depict our architecture which is trained by predicting discounted rewards obtained by trained reinforcement learning agents (see Appendix~\ref{appendix:predicting_rl}). We allow switching between training the policy and training the representation (i.e., selection of latent neurons). From training the policy to learning the representation, the training data changes only slightly. Importantly, in both cases, the data can  be created on-line by reinforcement learning agents.
}
\label{fig:rl_architecture_gpu}
\end{figure*}
In this appendix, we give the details for the architecture that has been used to factorise the abstract representation of a reinforcement learning environment. The code has been made available at \url{https://github.com/HendrikPN/reinforced_scinet}. 
For convenience, we repeat the training procedure here:
\begin{enumerate}
\item[(i)]{Train reinforcement learning agents.}
\item[(ii)]{Generate training data for representation learning from reinforcement learning agents (see Appendix~\ref{appendix:policy_data}).}
\item[(iii)]{Train encoders with decoders on training data to learn an abstract representation (see Appendix~\ref{appendix:proof_policy_data}).}
\end{enumerate}
The whole procedure is encompassed by a single algorithm (see Fig.~\ref{fig:rl_architecture_gpu}).

\subsection{Asynchronous reinforcement and representation learning}
Due to the highly parallelisable setting, we make use of asynchronous methods for reinforcement learning~\cite{babaeizadeh_2017_reinforcement}. That is, at all times, we have stored the neural network models in the shared memory of a graphics processing unit (GPU). Both, predicting and training, are therefore outsourced to the GPU while interactions of various agents with their environments are happening in parallel on central processing units (CPUs). The interface between the GPU and CPU is provided by two main processes which are assigned their own threads on CPUs, \emph{predictor}\footnote{Here we adopt the notation from Ref.~\cite{babaeizadeh_2017_reinforcement}. That is, the predictor processes used here are not related to the prediction process associated with decoders in the main text.} and \emph{training} processes. Predictor processes get observations from a \emph{prediction queue} and batch them in order to transfer them to the GPU where a forward pass of the deep reinforcement learning model is performed to obtain the policies (i.e., probability distributions over actions) which are redistributed to the respective agents. Training processes batch training data as appropriate for the learning model in the same way as predictors batch observations. This data is transferred to the GPU to update the neural network. In our case, we need to be able to switch between two such training processes. One for training a policy as in Ref.~\cite{babaeizadeh_2017_reinforcement} and as required by step~(i) of our training procedure, and one for representation learning as required by step~(iii). Interestingly, the training data which is used by the policy trainers in step~(i) is very similar to the training data which is used by the selection trainers in step~(iii). Therefore, in the transition from step~(i) to~(iii), we just have to slightly alter the data which is sent to the training queue as required by the algorithm in Sec.~\ref{appendix:policy_data}. Note that the similarity of the training data for the two training processes is due to the specific deep reinforcement learning model under consideration as described in the following section.
For further details on the implementation of asynchronous reinforcement learning methods on GPUs see Ref.~\cite{babaeizadeh_2017_reinforcement}.\\

 \begin{figure*}[ht!] 
\centering
\includegraphics[width=0.35\textwidth]{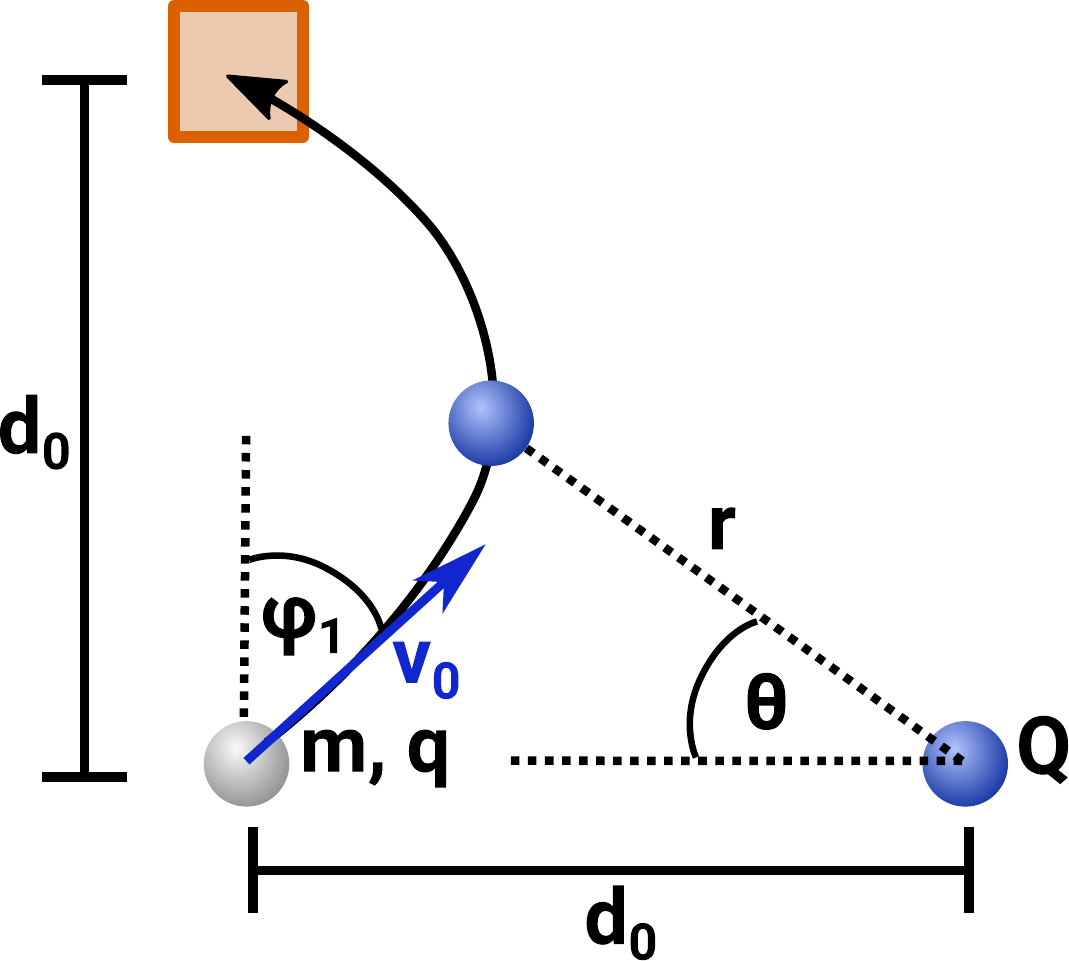} 
\caption{ \textbf{Setup and variable names for a charged mass being shot into a hole.} A charged particle with mass $m$ and charge $q$ moves in the electrostatic field generated by another charge $Q$ at a fixed position. The initial conditions are given by the velocity $v_0$ and the angle $\phi$. We want to determine the value for $\phi$ that will result in the particle landing in the target hole, given a velocity $v_0$.} 
\label{fig:charged_masses_appendix}
\end{figure*}

 \subsection{Deep energy-based projective simulation model}
The deep learning model used for the numerical results obtained here is a deep energy-based projective simulation (DPS) model as first presented in Ref.~\cite{jerbi_2019_framework}. We chose this model because it allows us to easily switch between training the policy and training the decoders since the training data is almost the same for both. In fact, besides different initial biases and network sizes, the models used as reinforcement learning agents and the models used for decoders are the same. 

The DPS model predicts so-called $h$-values $h(o,a)$ given an observation $o$ and action $a$.  The loss function aims to minimise the distance between the current $h$-value $h_t(o,a)$ and a target $h$-value $h_t^{\mathrm{tar}}(o,a)$ at time $t$, given as
\begin{align}
\mathcal{L}=| h_t(o,a) - h_t^{\mathrm{tar}}(o,a)|.\label{eq:training_dps}
\end{align}
Note that we are free to choose other loss functions such as the mean square error, or a Huber loss. 
We want the current $h$-value to be updated such that it maximises the future expected reward. Approximating this reward at time $t$ for a given discount factor, we write
\begin{align*}
R_t=\sum_{j=1}^{l_{\text{max}}}(1-\eta)^{j-1} r_{t+j}
\end{align*}
where $\eta\in(0,1]$ is the so-called \emph{glow} parameter accounting for the discount of rewards $r_{t+j}$ obtained after observing $o$ and taking action $a$ at time $t$ up to a temporal horizon $l_{\text{max}}$.
The target $h$-value can then be associated with this discounted reward as follows,
\begin{align*}
h_t^{\mathrm{tar}}(o,a)=(1-\gamma_{\text{PS}})h_t(o,a) + R_t,
\end{align*}
where $\gamma_{\text{PS}}\in [0,1)$ is the so-called \emph{forgetting} parameter used for regularisation.
The $h$-values are used to derive a policy through the softmax function,
\begin{align*}
\pi(a|o)=\frac{e^{\beta h(o,a)}}{\sum_{a'}e^{\beta h(o,a')}},
\end{align*}
where $\beta>0$ is an inverse temperature parameter which governs the drive for exploration versus exploitation. 
The tabular approach to projective simulation has been proven to converge to an optimal policy in the limit of infinitely many interactions with certain MDPs~\cite{clausen_2019_on} and has shown to perform as good as standard approaches to reinforcement learning on benchmarking tasks~\cite{melnikov2018benchmarking}. For a detailed description and motivation of the DPS model we refer to Ref.~\cite{jerbi_2019_framework}.

Note that the training data required to define the loss in Eq.\ (\ref{eq:training_dps}) consists of tuples containing observations, actions and discounted rewards $(o,a,R)$. Since this is in line with the training data required for training the decoders as described in Appendix~\ref{appendix:policy_data}, this model is particularly well suited for the combination with representation learning as introduced in this paper.

\section{Classical mechanics derivation for charged masses}
In this section, we provide the analytic solution to the charged masses example in Sec.~\ref{sec:charged_masses} that we use to evaluate the cost function for training the neural networks. This is a fairly direct application of the generic Kepler problem, but we include the derivation for the sake of completeness. We use the notation of Ref.~\cite{tong_classical_mechanics}.

The setup we consider is shown in Fig. \ref{fig:charged_masses_appendix}. Our goal is to derive a function $v_0(\phi)$ that, for fixed $q, Q, d_0$ and given $\phi$, outputs an initial velocity for the left mass such that the mass will reach the hole. Introducing the inverse radial coordinate $u = \frac{1}{r}$, the orbit $r(\theta)$ of the left mass obeys the following differential equation (see e.g., Ref.~\cite[Sec.~4.3]{tong_classical_mechanics}): 
\begin{equation} \label{eqn:orbit_ode}
\frac{d^2 u}{d \theta^2} + u = \frac{k}{l^2} \,,
\end{equation}
with the constant 
\begin{equation}
k = \frac{- q Q}{4 \pi \epsilon_0 m}
\end{equation}
and the mass-normalised angular momentum
\begin{equation}
l = r^2 \frac{d \theta}{d t} \,.
\end{equation}
This is a conserved quantity and we can determine it from the initial condition of the problem
\begin{equation}
l = d_0 v_0 \cos \phi \,.
\end{equation}
The general solution to Eq.~(\ref{eqn:orbit_ode}) is given by 
\begin{equation}
u = A \cos(\theta - \theta_0) + \frac{k}{l^2} \,,
\end{equation}
where $A$ and $\theta_0$ are constants to be determined from the initial conditions. The initial conditions are 
\begin{equation}
r(\theta = 0) = \frac{1}{A \cos(\theta_0) + \frac{k}{l^2}} = d_0 \,, \\
\end{equation}
\begin{equation}
\left. \frac{d r}{d \theta} \right\vert_{\theta = 0} = \frac{- A \sin\theta_0}{\left(A \cos \theta_0 + \frac{k}{l^2}\right)^2} \frac{v_0 \cos \phi}{d_0} = v_0 \sin \phi \,. 
\end{equation}
Combining these yields 
\begin{align}
A \cos \theta_0 &= \frac{1}{d_0} - \frac{k}{l^2} \,, \label{eqn:init_cond1} \\
A \sin \theta_0 &= - \frac{1}{d_0} \tan \phi \,. \label{eqn:init_cond2}
\end{align}
The condition that the mass reaches the hole is expressed in terms of $r(\theta)$ as follows: 
\begin{equation}
r\left( \theta = \frac{\pi}{4} \right) = \frac{1}{A \cos(\frac{\pi}{4} - \theta_0) + \frac{k}{l^2}} = \sqrt{2} d_0 \,.
\end{equation}
Using $\cos(\pi/4 - \theta_0) = \cos(\theta_0) / \sqrt{2} + \sin(\theta_0) / \sqrt{2}$ and the definition of $l$ as well as Eqs.~(\ref{eqn:init_cond1}) and~(\ref{eqn:init_cond2}), we can solve this for $v_0$: 
\begin{equation}
v_0^2 = \frac{(\sqrt{2} - 1) k}{d_0} \frac{1}{\cos \phi \sin \phi} \,.
\end{equation}

Restricting $\phi$ to a suitably small interval, this function is injective and has a well-defined inverse $\phi(v_0)$. The neural network has to compute this inverse from operational input data. To generate valid question-answer pairs, we evaluate $v_0(\phi)$ on a large number of randomly chosen $\phi$ (inside the interval where the function is injective).


\newpage
\bibliographystyle{stdWithTitle}
\bibliography{odr}

\end{document}